\documentclass[runningheads]{llncs}
\usepackage[T1]{fontenc}
\usepackage{graphicx}
\usepackage{hyperref}
\usepackage{color}

\usepackage{cite}
\usepackage{amssymb,amsmath,amsthm}
\usepackage{mathtools}
\usepackage{algorithmic}
\usepackage{textcomp}
\usepackage[dvipsnames]{xcolor}
\usepackage{tcolorbox}
\usepackage{url}
\usepackage{stmaryrd}
\usepackage{booktabs}
\usepackage{multicol}
\usepackage{multirow}

\usepackage{caption}
\usepackage{subcaption}

\usepackage[normalem]{ulem}

\usepackage{tikz}
\usepackage{tikzsymbols}
\usepackage{bbding}

\usepackage{microtype} % Improves spacing and reduces hyphenation
\usepackage[noend,linesnumbered,commentsnumbered,ruled]{algorithm2e}

\SetCommentSty{decmt}    
\allowdisplaybreaks

\SetKwInput{KwData}{Input}    
\SetKwInput{KwResult}{Output}    
\SetKw{KwBreak}{break}    
\SetKw{KwCont}{continue}    
\SetKw{KwContinue}{continue}    
\SetNlSty{textbf}{\color{black}}{}    
\SetNlSkip{1.1em}

\definecolor{darkgreen}{rgb}{0.0, 0.6, 0.13}

\usepackage{listings}
\definecolor{dkgreen}{rgb}{0,0.6,0}
\definecolor{gray}{rgb}{0.45,0.45,0.45}
\definecolor{mauve}{rgb}{0.58,0,0.82}
\definecolor{editorGray}{rgb}{0.92, 0.92, 0.92}
\definecolor{editorOcher}{rgb}{1, 0.35, 0} % #FF7F00 -> rgb(239, 169, 0)
\definecolor{editorGreen}{rgb}{0, 0.5, 0} % #007C00 -> rgb(0, 124, 0)

\lstdefinelanguage{JavaScript}{
morekeywords={typeof, new, true, false, catch, function, return, null, catch, switch, var, if, in, while, do, else, case, break, undefined},
morecomment=[s]{/*}{*/},
morecomment=[l]//,
morestring=[b]",
morestring=[b]'
}

\lstnewenvironment{JScriptnum}
{\lstset{%
	backgroundcolor=\color{editorGray},
	basicstyle={\scriptsize\ttfamily},   
	frame=l,
	xleftmargin={0cm},
	xrightmargin={0cm},
	numbers=left,
	stepnumber=1,
	firstnumber=1,
	numberfirstline=true,
	numberstyle = \color{darkgray}\ttfamily,
	keywordstyle=\color{blue}\ttfamily,
	commentstyle=\color{darkgray}\ttfamily,
	ndkeywordstyle=\color{editorGreen}\ttfamily,
	stringstyle=\color{editorOcher},
	language=JavaScript,
	alsodigit={.:;},
	tabsize=2,
	showtabs=false,
	showspaces=false,
	showstringspaces=false,
	columns=flexible,
	extendedchars=true,
	breaklines=true,
	breakatwhitespace=true
}}{}

\newcommand{\DOM}{\text{Dom}}
\newcommand{\IMG}{\text{Img}}

\newcommand{\LTL}{\mathsf{LTL}}
\newcommand{\LTLV}{\mathsf{LTL}_{\VarTrace}}
\newcommand{\HLTL}{\mathsf{HyperLTL}}
\newcommand{\sHL}{\mathsf{sHL}}
\newcommand{\shl}{$\mathsf{sHL}$\xspace}

\newcommand{\subst}{\varsigma}

\newcommand{\Quant}{\mathbb{Q}}

\newcommand{\modelsLTL}{\mathop{\models_{\texttt{\scriptsize LTL}}}}

\newcommand{\modelsHL}{\mathop{\models_{\texttt{\scriptsize HL}}}}

\newcommand{\alphabet}{\Sigma}
\newcommand{\domain}{\Sigma}
\newcommand{\Prop}{\mathbb{X}}

\newcommand{\Func}{\mathcal{F}}
\newcommand{\VarTrace}{\mathbb{V}}

\newcommand{\conc}{;}

\newcommand{\trace}{\tau}

\newcommand{\termSym}{\$}

\newcommand{\func}{\mathtt{f}}

\newcommand{\domainG}{\domain_{\func}}

\newcommand{\funcSign}{\sigma}

\newcommand{\funcTr}{\mathtt{id}_{\setTraces}}
\newcommand{\corrFI}{\mathop{\sqsubseteq_{\setTraces}}}

\newcommand{\Until}{\mathbin\mathbf{U}}
\newcommand{\Next}{\mathbf{X}}

\newcommand{\setTraces}{T}

\newcommand{\multiClass}{\mathcal{C}}
\newcommand{\denote}[1]{\llbracket #1 \rrbracket}

\newcommand{\Free}[1]{\textit{Free}(#1)}

\newcommand{\pref}{\mathop{\leq}}

\newcommand{\traceVar}{\pi}

\newcommand{\pv}[2]{{{\bf #1}(#2)}}
\newcommand{\pvF}[1]{{\bf #1}}

\newcommand{\sys}{\setTraces}
\newcommand{\obsQ}[2]{#1 #2\smallsp\mathop{::}\smallsp\sys}
\newcommand{\genQ}[3]{#1 #2\mathop{\in}#3}

\newcommand{\val}{v}
\newcommand{\va}[1]{\mathit{#1}}

\newcommand{\As}{\mathop{:}}

\newcommand{\traceAssign}{\Pi}

\newcommand{\traceprop}{\mathbf{T}}
\newcommand{\hyper}{\mathbf{H}}
\newcommand{\hyperSet}[1]{\hyper(#1)}
\newcommand{\powerSet}[1]{\mathcal{P}(#1)}
\newcommand{\fin}{\text{fin}}
\newcommand{\term}{\text{term}}

\newcommand{\boundPowerSet}[3]{\mathcal{P}^{#1}_{#2}(#3)}
\newcommand{\observ}{\mathcal{O}}

\newcommand{\extOrder}{\preceq}
\newcommand{\goodSet}{\mathrm{Good}}
\newcommand{\badSet}{\mathrm{Bad}}
\newcommand{\traceExtend}{\leq}
\newcommand{\hyperExtend}{\preceq}

\newcommand{\noSecret}{\mathtt{public}}

\newcommand{\ISO}{ISO\xspace}

\newcommand{\LIN}{linearizability\xspace}
\newcommand{\linear}{\mathtt{linear}}

\newcommand{\HC}{\mathtt{sec}}
\newcommand{\LC}{\mathtt{pub}}
\newcommand{\inVar}{\pvF{in}}
\newcommand{\outVar}{\pvF{out}}

\newcommand{\inlow}{\inVar_{\LC}}
\newcommand{\outlow}{\outVar_{\LC}}
\newcommand{\inhigh}{\inVar_{\HC}}

\newcommand{\phiLin}{\text{linear}}
\newcommand{\phieqOrd}{\text{eqOrder}}

\newcommand{\lab}{\pvF{label}}

\newcommand{\tAnd}{\text{ and }}

\newcommand{\tIf}{\text{ if }}

\newcommand{\tIff}{\text{ iff }}
\newcommand{\tSt}{\text{ s.t.\ }}

\newcommand{\Def}{\stackrel{\text{def}}{=}}
\newcommand{\nat}{\mathbb{N}}

\newcommand{\LLL}{\mathinner{\kern-0.1em.\kern-0.05em.\kern-0.05em.\kern-0.1em}}

\newcommand{\smallsp}{\kern-0.1em}

\label{key}%%% For collaboration
\definecolor{ana}{RGB}{0,100,200}   
\definecolor{mar}{RGB}{50,100,0}   
\definecolor{all}{RGB}{200,100,0}

\newcounter{notes}

\newcommand{\rev}[1]{{\color{blue} #1}}
\renewcommand{\rev}[1]{#1}

\makeatletter 
\def\orcidID#1{\smash{\href{http://orcid.org/#1}{\protect\raisebox{-1.25pt}{\protect\includegraphics{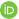}}}}}    
\makeatother

\begin{document}

% \bstctlcite is provided by IEEEtran.cls; define it here for llncs.
\makeatletter
\def\bstctlcite{\@bsphack\@bstctlcite}
\def\@bstctlcite#1{\@bsphack
  \@for\@citeb:=#1\do{%
    \edef\@citeb{\expandafter\@firstofone\@citeb}%
    \if@filesw\immediate\write\@auxout{\string\citation{\@citeb}}\fi}%
  \@esphack}
\makeatother
\bstctlcite{IEEEexample:BSTcontrol}

\title{Monitoring Hyperproperties over\\Observed and Constructed Traces}
\author{Marek Chalupa\inst{1}\thanks{Previously affiliated with ISTA when working on this paper.}\orcidID{0000-0003-1132-5516}
         \and
 Thomas A. Henzinger\inst{2}\orcidID{0000-0002-2985-7724} \and
 Ana Oliveira da Costa\inst{2}\orcidID{0000-0002-8741-5799}%
        \thanks{\Envelope~Corresponding author: \email{ana.costa@ista.ac.at}}
}
\authorrunning{M. Chalupa~et al.}

\institute{
Zeroth Research, Birmingham, United Kingdom\\\email{marek.chalupa@zeroth.org}\and
Institute of Science and Technology Austria (ISTA), Klosterneuburg, Austria\\
 \email{\{tah,ana.costa\}@ista.ac.at}
}

\maketitle              % typeset the header of the contribution

\begin{abstract}
\emph{Hyperproperties} specify system requirements that relate multiple executions. 
Reasoning about such properties naturally involves executions that play different roles: some correspond to concrete behaviors of the system under verification, while others represent idealized or alternative executions defined by the specification. 
We propose a framework that makes this semantic distinction explicit by distinguishing between \emph{observed} and \emph{constructed} traces.
Constructed traces are built using \emph{generators}: functions provided at specification time and invoked by monitors at runtime.
Our main motivation lies in the \emph{runtime verification} of hyperproperties, where only a finite view on possible executions is available. 
Our approach enables sound runtime verification for a broader class of hyperproperties than approaches that do not account for constructed traces, including properties with quantifier alternation, such as opacity in security, and properties defined with respect to idealized executions, such as linearizability in concurrency. 
We present a monitoring algorithm incorporating generator-constructed traces, and evaluate it on hyperproperties ranging from security to concurrency, demonstrating that constructed traces not only broaden expressivity but can also significantly accelerate violation detection.
\keywords{Hyperproperties  \and Monitoring \and Monitorability}
\end{abstract}

% ---- Introduction ----
\section{Introduction}
\label{sec:intro}

Many system requirements cannot be judged from a single execution; they define how multiple executions relate to one another. 
Such requirements are formally referred to as \emph{hyperproperties}~\cite{ClarksonS10}, and they arise naturally in a wide range of domains, including the specification of forbidden information flows in secure systems and consistency in concurrent systems.
Monitoring hyperproperties at runtime raises challenges that have no analogue in trace-property monitoring. 
In particular, 
a hyperproperty monitor's observed history may grow in two ways: new observations may extend an existing execution or start a new one;
and there may be no known bound on how many executions the system will eventually produce.
If the hyperproperty further requires the existence of specific executions, the monitor has no general way to know whether a still-missing witness will appear in some future observation or never appear at all. 
For such hyperproperties, no monitor can commit to a verdict from a finite observation of the system, i.e., the hyperproperty is not monitorable.

\begin{example}
The security property of \emph{initial-state opacity (ISO)} \cite{bryans2008opacity,opacityDES16} requires that every execution, as observed by users who see only the system's public information, is indistinguishable from some execution starting at a public state:
%
%\vspace{-1mm}
\begin{align}
	\obsQ{\forall}{\traceVar}\ \obsQ{\exists}{\traceVar'} \ (\lab(\traceVar')[0] \mathop{=} \LC \wedge \outlow(\traceVar) \mathop{=}\outlow(\traceVar')).
	\label{def:intro:opacity}
\end{align}
In the formula above, 
$\traceVar$ and $\traceVar'$ are \emph{trace variables} to be instantiated with traces from the system $\sys$ (denoted by
$\traceVar\smallsp\mathop{::}\smallsp\mathtt{\sys})$.
Here, $\outlow(\traceVar)$ is the sequence of all public output values in $\traceVar$ while
$\lab(\traceVar')[0]$ refers to the security label at the initial state of the execution assigned to $\traceVar'$ (which can either be $\LC$ for public executions or $\HC$ for secret ones).
A monitor for (\ref{def:intro:opacity}) could never flag a violation of \ISO because, at any moment, a new observation can add a trace witnessing the existential quantifier. 
Conversely, if we monitor for the satisfaction of \ISO, there is no guarantee that an update will not add a trace missing its witness.
Generally, hyperproperties with at least one quantifier alternation over the system's observations are not monitorable without further knowledge of the system~\cite{monitoringHyperLTL19,stucki2021gray}.
Now, 
suppose there exists a generator $\noSecret(\traceVar)$
that defines an over-approximation of all executions starting from public states.
We can use this generator to monitor for violations of \ISO; in particular, we monitor for
%
%\vspace{-1mm}
\begin{align}
\obsQ{\exists}{\traceVar}\, \genQ{\forall}{\traceVar'\!\!}{\noSecret(\traceVar)} \, (\lab(\traceVar'\smallsp)[0] \mathop{=} \LC \rightarrow \outlow(\traceVar)\mathop{\neq}\outlow(\traceVar')).
\label{def:intro:opacity:gen}
\end{align}
Constructed traces are drawn from a given generator, which we make explicit in the specification above by writing $\traceVar'\mathop{\in}\noSecret(\traceVar)$.
Given the formula above, instead of waiting for missing traces, a monitor can now use the generator $\noSecret(\traceVar)$ to check for their existence.
We observe that
all systems over-approximated by $\noSecret$ can be soundly monitored for violations of (\ref{def:intro:opacity}) using (\ref{def:intro:opacity:gen}).
\hfill $\blacktriangleleft$
\end{example}
%\vspace{-2mm}

A second challenge arises from a less-explored aspect of hyperproperty specifications: they may refer to executions the system will never produce, going beyond what any observation of the system can provide.
Examples include reasoning over sequential reordering of concurrent histories \cite{LinearHerlihyWing90}, 
leak-free variants of an observed run \cite{mcleanInterleavingFunc94}, 
and abstractions representing an attacker's view \cite{ANI04,ANIJournal18}.
Such specifications are prevalent across security, concurrency, and other domains, yet existing runtime verification approaches for hyperproperties consider only specifications that quantify over executions produced by the system. 

\begin{example}
\emph{Linearizability} \cite{LinearHerlihyWing90} is a consistency property for concurrent data structures, requiring that all histories of invocations and responses to shared resources be consistent with a sequential execution of those operations. 
When verifying \LIN,
all that matters is that the implementation's histories can be successfully mapped to the ``idealized'' sequential implementation.
A possible abstract specification of \LIN, using the distinction proposed in \cite{flavorsFOL25} between trace quantifiers constrained and unconstrained by a system, is:
%
%
%\vspace{-2mm}
\begin{flalign}
	\obsQ{\forall}{\traceVar}\   \exists \traceVar' \ (\varphi_{\phiLin}(\traceVar') \wedge    
	\varphi_{\phieqOrd}(\traceVar,\traceVar')).
	\label{def:intro:linear}
\end{flalign}
The formula above starts with a constrained universal quantifier ranging over observed traces, which are sequences of operations on a monitored shared object (i.e., our observed executions over the system $\sys$), followed by an existential quantifier ranging over all possible traces.
The inner formula \(\varphi_{\phiLin}(\traceVar')\) specifies that \(\traceVar'\) defines a linear trace, while \(\varphi_{\phieqOrd}(\traceVar,\traceVar')\) requires the linear trace \(\traceVar'\) to be an equivalent re-ordering of the observed one. %  while preserving the precedence order defined by~\(\traceVar\).
%
%\anaSide{It is not even monitorable.}
%If \LIN were specified using only quantifiers constrained by system observations, monitoring could yield false negatives: the system’s failure to produce a witness for the existential quantifier does not necessarily imply that no linear history matches the current observation.

Because $\traceVar'$ does not come from the system, the specification (\ref{def:intro:linear}) is {in fact} monitorable despite the quantifier alternation: each observation $\traceVar$ consists of a finite number of operations, so only a finite number of reordered traces need to be considered for $\traceVar'$.
However, considering all possible traces is not optimal.
To improve the monitor, {we can consider $\traceVar'$ not to be arbitrary, but to be taken only from the set \(\linear(\traceVar)\) of all legal linear histories over the operations in $\traceVar$}:
%\vspace{-2mm}
\begin{flalign}
	\obsQ{\forall}{\traceVar}\  \genQ{\exists}{\traceVar'}{ \linear(\traceVar)}\ \ 
	(\varphi_{\phiLin}(\traceVar') \wedge    
	\varphi_{\phieqOrd}(\traceVar,\traceVar')).
	\label{def:intro:linear:generated}
\end{flalign}
The specification above fits the usual definition of the monitoring problem for linearizability \cite{testingWing93,testingLowe17}, which is parameterized by the sequential specification of an abstract data type (ADT).
In this view, the \emph{generator function} \(\linear\) is the %sequential
ADT specification.
We observe that \(\linear\) specifies the requirements on the generator function to be used at runtime: any generator that over-approximates \(\linear\) can soundly be used to monitor for violations of (\ref{def:intro:linear:generated}).
\hfill $\blacktriangleleft$
\end{example}

Both challenges share a common trait: effective monitoring of hyperproperties can benefit from traces that are not directly observed from the system but can instead be constructed. 
We address this by extending hyperproperty specifications with \emph{generators}: semantic objects given as part of the specification that provide a principled way to extend monitoring beyond passive observation of systems.
% 
%Our approach unifies mechanisms introduced in prior work, such as secure multi-execution \cite{SME16}, testable hypersafety policies \cite{blackenforc15}, and gray-box monitoring~\cite{stucki2021gray}.

\medskip\noindent
\textbf{Contributions.} This paper explores the theory % %, foundations,
and implementation of runtime verification of hyperproperties whose specifications mix executions observed from the system with executions constructed by the specification itself. 
Our contributions are as follows\footnote{Proofs are available in the extended version of this paper~\cite{arxiv-ext}.}.

\begin{itemize}
	\item \textbf{Specifications over observed and constructed traces.}
	We make explicit, at the semantic level, the distinction between executions
	obtained from the system and executions introduced by the specification,
	modeling the latter as \emph{generator functions}.
	
	\item \textbf{Sound replacement of observation by generation.}
	We give criteria characterizing when a
	generator may soundly replace quantification over observed and constructed traces.
	Such replacements extend runtime verification to a broad class of hyperproperties (those
	with quantifier alternation and those defined against idealized
	executions) and let generators to soundly accelerate the detection of violations.
	
	\item \textbf{Monitoring algorithm and evaluation.}
	We give an algorithm for monitoring linear-time hyper specifications with both
	observing and generating quantifiers that reuses existing monitors for their
	alternation-free fragment.
	%
	%Building on existing $\HLTL$~\cite{monitorHyper17} and Hypernode Logic~\cite{hypernodeMonInfinite25} monitors,
    We show empirically that generators enable monitoring under quantifier alternation and
	improve the monitoring of universal hyperproperties.
\end{itemize}

\section{Observed and Constructed Traces}
\label{sec:obs_constr}

We represent \emph{program executions as finite sequences of symbols} from an alphabet \(\alphabet\), possibly terminated with the symbol \(\termSym \mathop{\not\in} \alphabet\).
We refer to executions as \emph{traces} (i.e., 
a trace is \(\trace \mathop{\in} \alphabet^* \{\termSym\}^*\)), and
we may refer to traces ending with \(\termSym\) as \emph{terminated traces}.
Let \(\trace\mathop{=}\val_0\ldots \val_n \termSym^m\) be a trace where \(\val_i\mathop{\in}\alphabet\) and \(n,m\mathop{\in}\nat\). 
Its 
\emph{size} is the index of the last valuation, \(|\trace|\mathop{=}n+1\).
For traces \(u\in\alphabet^*\) and \(v\in\alphabet^* \{\termSym\}^*\), we write their \emph{concatenation} as \(u \conc v\) or simply \(uv\).
We generalize concatenation to sets of traces as usual.
For a  trace \(\trace \mathop{=} uv\), we say that \(u\) is a \emph{prefix of} \(\trace\), denoted \(u \pref \trace\), and \(v\) is a \emph{suffix of}  \(\trace\).
For terminated traces \(u\) and \(\trace\): \(u \pref\trace\) iff 
\(\trace \mathop{\in} \{u\}\{\termSym\}^*\).
For \(i\mathop{\leq}n\), the \emph{suffix of \(\trace\) from position \(i\)} is \(\trace[i\LLL]\mathop{=}\val_i\ldots \val_n \termSym^m\) and the symbol at position \(i\) is \(\trace[i]\mathop{=}\val_i\).
For \(i\mathop{>}n\), we set \(\trace[i\LLL]\mathop{=}\epsilon\), and \(\trace[i]\mathop{=}\termSym\) if \(0\mathop{<}m\) and \(\trace[i]\mathop{=}\epsilon\) otherwise.
The notion of trace prefix is lifted to sets of traces~\cite{ClarksonS10}: \(\setTraces\) is a prefix of \(\setTraces'\), written \({\setTraces \hyperExtend \setTraces'}\), iff
\(\forall \trace \mathop{\in} \setTraces\, \exists \trace'\mathop{\in} \setTraces' \ \trace \mathop{\traceExtend} \trace'\).
Henceforth,
we may leave the symbol $\termSym$ implicit (i.e., we may use $\domain^*$ to denote $\alphabet^* \{\termSym\}^*$).

A \emph{trace property} $\traceprop$ is a set of traces, satisfied by an execution $\trace$ iff $\trace \mathop{\in} \traceprop$. 
A \emph{system} is likewise a set of traces (i.e., the set of all its possible executions).
\emph{Hyperproperties} require comparing multiple executions and therefore cannot be expressed as trace properties \cite{ClarksonS10}.
A \emph{hyperproperty} \(\hyper\) defines a set of sets of traces, 
so a system $\sys$ satisfies a hyperproperty $\hyper$ when it belongs to it (i.e., $\sys\mathop{\in}\hyper$).

% ---- Generators ----
%\subsection{Generators for Constructed Traces}
%\label{sec:gen}

%Before introducing our specification language in the next section, we formalize the notion of \emph{generators}.
%%
%Generators construct traces that can be used for comparison in relational specifications.

\smallskip
\noindent\textbf{Generators.}
To support incremental construction of traces, we define generators over finite prefixes. 
We distinguish between two execution domains: the \emph{observed domain} $\domain$ for executions that are consumed by the generator,
and $\domainG$ for the \emph{domain of traces constructed by the generator $\func$}.
Let us write \(\powerSet{\setTraces}\) for the set of all subsets of \(\setTraces\).
\emph{Generators} are $n$-ary functions on trace prefixes:
$\func\mathop{:}(\domain^*)^n \rightarrow \boundPowerSet{}{}{\domainG^*}$.
We say that a generator $\func$ is \emph{domain-preserving} when $\domain_\func \mathop{=} \domain$, i.e.\ it
constructs traces over the same alphabet it observes. 
The arity $n$ of a generator is fixed and reflects the number of executions it depends on.
Note that $n = 0$ is allowed, that is, a generator can depend on \emph{no} executions.
We require generators to be \emph{monotonic}; that is, for all tuples of trace prefixes,
if $(\trace_1, \LLL, \trace_n) \mathop{\leq} (\trace_1', \LLL, \trace_n')$
then $\func(\trace_1, \LLL, \trace_n) \mathop{\extOrder} \func(\trace_1', \LLL, \trace_n')$,
where $(\trace_1, \LLL, \trace_n) \mathop{\leq} (\trace_1', \LLL, \trace_n')$ requires that $\trace_i\mathop{\leq}\trace'_i$ for all $i\mathop{\leq}n$.
A \emph{generator signature}, referred to as \(\Func\), is a set of pairs \((\func,n)\) where 
\(\func\) is a function symbol and \(n\) is its  arity.
We represent a concrete \emph{interpretation of generators in a signature \(\Func\) over the alphabets \(\domain\) and $\domainG$} as:
$
\funcSign\mathop{=}((\domain,\{\domainG \mid (\func,n)\mathop{\in} \Func\}), \{\func^{\funcSign}\As (\domain^*)^n  \rightarrow \boundPowerSet{}{}{\domainG^*}\mid (\func,n)\mathop{\in} \Func\}).
$

\smallskip
\noindent\textbf{Approximating Systems with Generators.}
One use of generators in runtime verification is to soundly replace the
set of traces $\setTraces$ defined by the monitored system with an approximation of it. 
The soundness results in Section~\ref{sec:correct_active_spec} build on this input-level approximation: on the traces the generator produces from inputs in $\setTraces$, an over-approximation must define a superset of $\setTraces$ so that nothing required by a universal quantifier is lost; an under-approximation must stay within $\setTraces$ so that nothing spurious is admitted under an existential one.
%We formalize this below and use it in
%Section~\ref{sec:correct_active_spec} to establish soundness criteria for monitoring.

\begin{definition}\label{def:gen:approx}
Let $\setTraces \mathop{\subseteq} \domain^*$ be a set of traces,
$\func \mathop{:} (\domain^*)^n \to \mathcal{P}(\domain^*)$ a 
(domain-preserving) 
generator of arity $n$ and $\funcSign$ a generator interpretation.
We say that $\func^{\funcSign}$ \emph{under-approximates}~$\setTraces$,
denoted $\func^{\funcSign} \mathop{\sqsubseteq} \setTraces$, when
$
\bigcup \big\{\func^{\funcSign}(\trace_1,\LLL,\trace_n) \mid
(\trace_1,\LLL,\trace_n) \mathop{\in} \setTraces^n \big\}$ $\mathop{\subseteq} \setTraces,
$
and that $\func^{\funcSign}$ \emph{over-approximates} $\setTraces$, denoted $\setTraces \mathop{\sqsubseteq} \func^{\funcSign}$, when
\(
\setTraces \mathop{\subseteq} \bigcap \big\{\func^{\funcSign}(\trace_1,\LLL,\trace_n) \mid
(\trace_1,\LLL,\trace_n) \mathop{\in} \setTraces^n \big\} .
\)
\end{definition}

\section{Hyperproperties with Generators}
\label{sec:hyper:spec}
The distinction between observed and constructed executions is orthogonal to the formalism used to specify hyperproperties.
We build on the abstract framework of~\cite{DBLP:conf/lics/BozzelliPS21}, which captures a wide range of formalisms for linear-time \mbox{hyperproperties}, by adding an explicit distinction between \emph{observed}  and \emph{constructed} traces.
We restrict attention to finite traces, working with \emph{pointed-traces} $(\trace,i)$ that pair a trace $\trace$ with a time point $i\mathop{\in}\nat$.

\medskip
\noindent\textbf{Multi-trace Specifications.}
Linear-time hyper specifications are built from \emph{multi-trace specifications} \cite{DBLP:conf/lics/BozzelliPS21} that formalize relations among traces without including trace quantification.
Formally, a multi-trace specification, $\psi(\traceVar_1, \ldots, \traceVar_n)$, 
is parameterized by its set of free trace variables (i.e., $\Free{\psi}\mathop{=}\{\traceVar_1, \ldots, \traceVar_n\}$) and it is interpreted over pointed-trace assignments.

A \emph{({pointed}-trace) assignment} $\traceAssign$ is a partial mapping from a  set $\VarTrace$ of trace variables to a trace and a time point:  $\traceAssign\As \VarTrace \rightarrow (\alphabet^*\times \nat)$.
We define the \emph{set of all trace variables assigned in $\traceAssign$} as $\DOM(\traceAssign)\mathop{=}\{\traceVar \mid \exists (\trace,i) \ \traceAssign(\traceVar)\mathop{=}(\trace,i)\}$ and the \emph{set of all pointed-traces assigned by $\traceAssign$} as  $\IMG(\traceAssign)\mathop{=}\{\traceAssign(\traceVar) \mid \traceVar\mathop{\in}\DOM(\traceAssign)\}$.
An assignment is \emph{initial} when all its time points are equal to~$0$.
The semantics of a multi-trace specification, {denoted} $\denote{\psi(\traceVar_1, \ldots, \traceVar_n)}$, is given as a set of initial pointed-trace assignments,
and we say that an \emph{assignment $\traceAssign$ satisfies $\psi(\traceVar_1, \ldots, \traceVar_n)$} iff 
$\traceAssign\mathop{\in}\denote{\psi(\traceVar_1, \ldots, \traceVar_n)}$.
The next two examples introduce the multi-trace specifications used  to specify hyperproperties in this work.

%\vspace{-1.5mm}
\begin{example}
Let $\Prop_{\VarTrace}\mathop{=}\{x_{\traceVar}\mid x\mathop{\in}\Prop \tAnd \traceVar\mathop{\in}\VarTrace\}$ be a (finite) set of propositional variables obtained by annotating propositional variables in $\Prop$ with trace variables in $\VarTrace$.
The formulas of $\LTLV$ are defined by the grammar
$\psi\!::=\! x_{\traceVar} \, |\, \neg \psi \, |\, {\psi \vee \psi} \, | \, \Next \psi \, | \, {\psi\! \Until\! \psi}$
where $x_{\traceVar}\mathop{\in}\Prop_{\VarTrace}$.
We define multi-trace specifications over $\LTLV$ formulas by adapting the semantics of linear temporal logic ($\LTL$) on finite traces~\cite{LTLfinite13} to be evaluated over assignments with alphabet $\alphabet\mathop{=}\powerSet{\Prop}$.

For an assignment $\traceAssign$ and $i\mathop{\in}\nat$, we denote by $(\traceAssign\mathop{+}i)$ the update of $\traceAssign$ that increments all time points by $i$ and we say that \emph{$(\traceAssign\mathop{+}i)$ is defined}, denoted $(\traceAssign\mathop{+}i)\mathop{\downarrow}$, when the new time points do not exceed the size of any trace in the image of $\traceAssign$.
Formally, $(\traceAssign\mathop{+}i)\mathop{\downarrow}$ iff, for all $(\trace,j)\mathop{\in}\IMG(\traceAssign)$, we have $j\mathop{+}i\mathop{\leq}|\trace|$.
We define inductively that an assignment $\traceAssign$ satisfies an $\LTLV$ formula, as follows (where boolean connectives are defined as usual):
\begin{align*}
	&\traceAssign \modelsLTL x_{\traceVar} \tIff x\mathop{\in}\trace[i] \text{ where } (\trace,i)\mathop{=}\traceAssign(\traceVar) \\
	&\traceAssign \modelsLTL \Next  \psi \tIff (\traceAssign\mathop{+}1) \mathop{\downarrow} \tAnd (\traceAssign\mathop{+}1) \modelsLTL  \psi\\
	&\traceAssign \modelsLTL  \psi_1\! \Until\!  \psi_2 \tIff
	\text{exists } j\mathop{\geq} 0 \tSt
	(\traceAssign\mathop{+}j)\mathop{\downarrow} \tAnd
	(\traceAssign\mathop{+}j) \modelsLTL  \psi_2,\\
	&\quad
	\text{and for all } 0\mathop{\leq} j'\mathop{<}j,
    (\traceAssign\mathop{+}j') \modelsLTL  \psi_1.
\end{align*}
The $\LTLV$ formula $\psi$ defines the multi-trace specification with interpretation $\denote{\psi}\mathop{=}\{\traceAssign\mid {\traceAssign \text{ is initial}} \tAnd \traceAssign\modelsLTL\psi\}$.
\hfill $\blacktriangleleft$
\end{example}

\newcommand{\lt}{a}
\begin{example}
\label{ex:od:hypernode}
Consider a fragment $\sHL_{\leq}$ of \emph{symbolic hypernode logic ($\sHL$)}~\cite{hypernodeMonInfinite25}, with formulas defined by the grammar
$
\psi ::= \neg \psi \mid \psi \wedge \psi \mid \pvF{x}(\traceVar) \leq \pvF{x}(\traceVar) \mid  \pvF{x}(\traceVar)  \mathop{\precsim} \pvF{x}(\traceVar) 
$
where  $\traceVar$ is a trace variable, $\pvF{x}$ is the projection of a trace onto the values of the variable $x\in\Prop$, and $\leq$ and $\precsim$ are ordinary and stuttering, resp., prefixing relations.
The alphabet is the set of all mappings from variables in $\Prop$ to a (possibly infinite) domain $D$; that is,  $\alphabet\mathop{=}{D}^{\Prop}$.
%Thus, a trace represents a sequence of states where variables have values from $D$.
%and projection $\pvF{x}(\traceVar)$ projects trace $\traceVar$ on the sequence of values of variable $x\in\Prop$, which is a sequence living in $D^*$.
%For example, if $D = \nat$ and $\Prop = \{x, y\}$, then $\pvF{x}(\traceVar)$ might be $1\cdot 2 \cdot 3$.
For a trace $\tau = \val_0 \ldots \val_n \in \alphabet^*$ and $x \in\Prop$, we denote $\tau(x)$ the (element-wise) projection to values of $x$, that is, $\tau(x) = \val_0(x) \ldots \val_n(x)$.

We define that 
$\traceAssign$ satisfies the atomic formulas of $\sHL_{\leq}$ as follows:
\begin{align*}
	&\traceAssign \modelsHL \pvF{x}(\traceVar)\!\mathop{\leq}\! \pvF{x'}(\traceVar')  \tIff \trace[i\LLL](x)\mathop{\leq} \trace'[i'\LLL](x')
	\text{ for } (\trace,i)\mathop{=}\traceAssign(\traceVar) \tAnd (\trace',i')\mathop{=}\traceAssign(\traceVar')\\
	&\traceAssign \modelsHL \pvF{x}(\traceVar)  \mathop{\precsim} \pvF{x'}(\traceVar')  \tIff 
	\trace[i\LLL](x)\mathop{\in}\lt_0^+\LLL \lt_n^+ \tAnd\trace'[i'\LLL](x')\mathop{\in}\lt_0^+\LLL \lt_n^+ D^*, \\
	&\hspace{25mm} \text{for } \lt_j \neq \lt_{j+1},  0\leq j < n, \text{ and }(\trace,i)\mathop{=}\traceAssign(\traceVar) \tAnd (\trace',i')\mathop{=}\traceAssign(\traceVar').
\end{align*}
In words, if $\traceAssign$ maps
$\traceVar$ to $(\tau, i)$ and $\traceVar'$ to $(\tau', i')$, then
$\pvF{x}(\traceVar) \leq \pvF{x'}(\traceVar')$ holds when $\tau[i\LLL](x)$ is a prefix of $\tau'[i'\LLL](x')$, and 
$\pvF{x}(\traceVar) \precsim \pvF{x'}(\traceVar')$ holds when $\tau[i\LLL](x)$ is a prefix of $\tau'[i'\LLL](x')$ after removing stuttering (repetitions) from $\tau[i\LLL](x)$ and  $\tau'[i'\LLL](x')$.
The boolean connectives are defined as usual.
%
%We define the multi-trace specification (under the semantics above) for the quantifier-free part of OD~as
%
%\begin{align*}
%	(\inlow(\traceVar) \mathop{=} \inlow(\traceVar')
%	\rightarrow
%	\outlow(\traceVar) \mathop{=}\outlow(\traceVar'))
%\end{align*}
%where $\inlow(\traceVar) \mathop{=} \inlow(\traceVar')$ is an abbreviation for
%$\inlow(\traceVar) \mathop{\leq} \inlow(\traceVar') \wedge \inlow(\traceVar') \mathop{\leq} \inlow(\traceVar)$. 
\hfill $\blacktriangleleft$
\end{example}

\noindent\textbf{Quantifying over Observed and Constructed Traces.}
We consider two types of quantifiers:
\emph{observing quantifiers} ranging over observed traces (e.g.,  $\obsQ{\exists}{\traceVar}\ \varphi$), and
\emph{\mbox{generating} quantifiers} ranging over traces constructed by a generator (e.g., $\genQ{\exists}{\traceVar}{\func(\traceVar_1, \ldots, \traceVar_n)}\ \varphi$).
The syntactic separation between generating and observing quantifiers distinguishes information available at specification time (generators) from information that may only be available at verification time (the behavior of the system under verification).

\begin{definition}
Let $\multiClass$ be a class of multi-trace specifications.
A \emph{linear-time hyper formula over $\multiClass$} with both observing and generating quantifiers is defined by the grammar:
$
\varphi \mathbin{::=}  \obsQ{\exists}{\traceVar}\ \varphi \mid   \genQ{\exists}{\traceVar}{\func(\traceVar, \ldots, \traceVar)}\ \varphi \mid \neg \varphi \mid  \psi(\traceVar, \ldots, \traceVar)
$
where $\sys$ is a constant,
$\traceVar$ is a trace variable, $\psi(\traceVar, \ldots, \traceVar)$ is a multi-trace formula from $\multiClass$ 
and
$\func$ is a generator function symbol from a generator signature $\Func$.
\end{definition}

The set of all generator terms over $\Func$ and variables in $\VarTrace$
is denoted by 
$
\mathbb{F}(\Func,\VarTrace).
$
As usual, $\obsQ{\forall}{\traceVar}\ \varphi \mathop{\Def} \neg \obsQ{\exists}{\traceVar} \neg\varphi$ and analogously for the universal generating quantifier.
We say that a trace variable is \emph{bound} if it is quantified.
We refer to the set of all bound variables in $\varphi$ as $\textit{Bound}(\varphi)$, while the set of all free variables (not bound) in $\varphi$ is $\Free{\varphi}$.
Without loss of generality, we assume that trace variables are quantified only once in a formula.
We also assume that generators are not defined over variables bound in their scope; that is, we write
$\genQ{\exists}{\traceVar}{\func(\traceVar_1, \ldots, \traceVar_n)} \ \varphi$ only if $\{\traceVar_1, \ldots, \traceVar_n\}\mathop{\cap} \textit{Bound}(\varphi)\mathop{=}\emptyset$.
In the following, $\Quant\mathop{\in}\{\forall, \exists\}$ and
$\VarTrace\mathop{=}\{\traceVar_1, \LLL, \traceVar_n\}$ is a set of trace variables.
We refer to linear-time hyper formulas as hyper formulas.

We observe that 
generators are defined over traces (Section~\ref{sec:obs_constr}), while assignments here map variables to pointed-traces. 
In the definition below, we apply a generator to pointed-trace arguments by projecting each onto its trace component\footnote{One could instead type generators to receive and produce pointed-traces.
There is, however, no canonical way to do this lift: for different formalisms and applications $(\trace,i)$ could be interpreted, for example, as $\trace$, its prefix, or its suffix $\trace[i\LLL]$.
}.
Formally, from now on, we interpret $\func^{\funcSign}(\traceAssign(\traceVar_1), \LLL, \traceAssign(\traceVar_n))$ as $\func^{\funcSign}(\trace_1, \LLL, \trace_n)$ where $\traceAssign(\traceVar_j) \mathop{=}(\trace_j, i_j)$, for $j\leq n$. 
Let $\setTraces$ be a set of traces, $\traceAssign$ an assignment, and  $\funcSign$ an interpretation for generators in $\Func$.
We define inductively 
$(\traceAssign, {\setTraces},\funcSign) \mathop{\models} \varphi$ for hyper specifications \(\varphi\) with generators interpreted with $\funcSign$ and multi-specifications in $\multiClass$ as:
\vspace{-3mm}
\begin{align*}
%---- Quantifiers - System
&(\traceAssign, {\setTraces},\funcSign) \models \obsQ{\exists} {\traceVar}\ \varphi 
\tIff
\text{exists }
\trace \mathop{\in} \setTraces \tSt
(\traceAssign[\traceVar \mapsto (\trace,0)], {\setTraces}, \funcSign) \models \varphi;\\
%---- Quantifiers - Functions
&(\traceAssign, {\setTraces},\funcSign) \models \genQ{\exists} {\traceVar}{\func(\traceVar_1, \LLL, \traceVar_n)}\ \varphi 
\tIff
\\
&
\hspace{20mm}
\text{exists }
\trace \mathop{\in} \func^{\funcSign}(\traceAssign(\traceVar_1), \LLL, \traceAssign(\traceVar_n)) \tSt
(\traceAssign[\traceVar \mapsto (\trace,0)], {\setTraces}, \funcSign)  \models \varphi;\\
%---- Negation
%
&(\traceAssign, {\setTraces},\funcSign)  \mathop{\models}  \neg  \varphi
\tIff 
(\traceAssign, {\setTraces},\funcSign)  \mathop{\not\models}  \varphi ;\\
%---- Mutli-trace
%\hspace{2mm}
&(\traceAssign, {\setTraces},\funcSign)  \mathop{\models}  \psi(\traceVar_1, \ldots, \traceVar_n)
\tIff 
\traceAssign\mathop{\in}\denote{\psi(\traceVar_1, \ldots, \traceVar_n)}_{\multiClass}
\end{align*}
where $\traceAssign[\traceVar\mathop{\mapsto}(\trace,i)]$ defines an update of the assignment where the value of $\traceVar$ changes to $(\trace,i)$ and, otherwise, remains the same
and $\denote{\cdot}_{\multiClass}$ is the multi-trace interpretation defined for $\multiClass$.
A set of traces $\setTraces$ and interpretation $\funcSign$ satisfy a hyper specification $\varphi$, denoted ${(\setTraces,\funcSign) \models \varphi}$ iff there exists an assignment $\traceAssign$  s.t., for all $\traceVar\mathop{\in}\Free{\varphi}$ and $(\trace,i)\mathop{=}\traceAssign(\traceVar)$, $\trace\mathop{\in}\setTraces$ and $(\traceAssign, \setTraces, \funcSign)\models \varphi$.
For formulas without generators, we may drop the interpretation $\funcSign$ and write only ${\setTraces\models \varphi}$.
We denote by $\hyperSet{\varphi, \funcSign} \mathop{=} \{\setTraces\conc  \{\termSym\}\mid (\setTraces,\funcSign) \mathop{\models} \varphi\}$ the hyperproperty defined by $\varphi$ and generator interpretation~$\funcSign$.

\vspace{-2mm}
\begin{example}
The hyper specification defined over the class of $\LTLV$ multi-trace specifications using only observing quantifiers is a possible interpretation of $\HLTL$ over finite traces \cite{hyperltl14}, while the one defined from $\sHL_{\leq}$ formulas is a fragment of symbolic hypernode logic \cite{hypernodeMonInfinite25}.
\hfill $\blacktriangleleft$
\end{example}
\begin{example}
\emph{Observational determinism} (OD) requires that any two system executions with the same public inputs produce the same public outputs.
We specify it using $\sHL_{\leq}$ as the linear-time hyper specification:
\begin{align}
	\obsQ{\forall}{\traceVar}\ \obsQ{\forall}{\traceVar'}\
	(\inlow(\traceVar) \mathop{=} \inlow(\traceVar')
	\rightarrow
	\outlow(\traceVar) \mathop{=}\outlow(\traceVar'))
\end{align}
where $\pvF{y}(\traceVar) \mathop{=} \pvF{y}(\traceVar')$ is
$\pvF{y}(\traceVar) \mathop{\leq} \pvF{y}(\traceVar') \wedge \pvF{y}(\traceVar') \mathop{\leq} \pvF{y}(\traceVar)$, for $\pvF{y}$ in $\{\inlow, \outlow\}$.
\hfill $\blacktriangleleft$
\end{example}

% ---- Relation between active and passive specifications ----
\section{Sound Replacement of Quantifiers}
\label{sec:correct_active_spec}

In settings such as the runtime verification of hyperproperties,
it may be useful to replace observations with trusted information about the system's possible behaviors.
In our setting, this is achieved by replacing quantifiers over observed executions with quantifiers ranging over traces constructed by a generator.
A natural question is then when such a replacement is sound: \emph{if the system satisfies the generator-based specification, does it also satisfy the original specification?} 
%
%We address this question by first looking into simple approximation-based substitutions of observed by constructed traces, followed by a generalization of this criterion to substitutions between generators.
%
%Proofs are provided in the extended version.

\smallskip
\noindent\textbf{Replacing Observing by Generating Quantifiers Soundly.}
When replacing observing quantifiers, soundness can be characterized directly in terms of over- and under-approximation of the system under verification.
We formalize this intuition by introducing the notion of generator substitution and then establishing sufficient conditions for its soundness.
In what follows, we denote by $\overline{\traceVar}$ a tuple of trace variables and by $\traceVar$ a single trace variable, and
use $\traceVar_0 \mathop{\in} \overline{\traceVar}$ to say that $\traceVar_0$ occurs in $ \overline{\traceVar}$.
We also fix $\setTraces$ to be a set of traces, $\funcSign$ to be a generator interpretation and $\Func$ to be a generator signature.

\begin{definition}
A \emph{generator substitution \(\subst\)}, for trace variables in \(\VarTrace\) and generators in $\Func$, is a partial function from variables $\traceVar_0 \mathop{\in} \VarTrace$ to generator terms $\func(\overline{\traceVar})\mathop{\in}\mathbb{F}(\Func,\VarTrace)$, s.t.\ $\traceVar_0\mathop{\notin} \overline{\traceVar}$.
The \emph{application of $\subst$ to a hyper formula $\varphi$} is defined inductively:
\begin{equation*}
\varphi[\subst]\!\mathop{=}\!
\begin{cases}
    \genQ{\Quant}{\traceVar_0}{\func'(\overline{\traceVar}')}  (\varphi'[\subst]) 
    &\hspace{-3mm} \tIf  
    \varphi\mathop{\in}\{\genQ{\Quant}{\traceVar_0}{\func(\overline{\traceVar})}  \varphi',
    \obsQ{\Quant}{\traceVar_0} \  \varphi'\} \tAnd
    \subst[\traceVar_0] \mathop{=} \func'(\overline{\traceVar}')
    %\\ &\ \ 
    \\
    \obsQ{\Quant}{\traceVar_0}\  (\varphi'[\subst]) 
    &\hspace{-3mm} \tIf  
    \varphi\mathop{=}\obsQ{\Quant}{\traceVar_0}\ \varphi' 
    %\\&\ \ 
    \tAnd \subst[\traceVar_0] \text{ not defined}\\
    \genQ{\Quant}{\traceVar_0}{\func(\overline{\traceVar})}\  (\varphi'[\subst])
    \hspace{-4mm} 
    &\hspace{-3mm} \tIf  
    \varphi\mathop{=}\genQ{\Quant}{\traceVar_0}{\func(\overline{\traceVar})}\,\varphi' 
    %\\&\ \ 
    \tAnd \subst[\traceVar_0] \text{ not defined}\\
    \psi(\traceVar_1,\LLL,\traceVar_n) 
    &\hspace{-3mm} \tIf \varphi=\psi(\traceVar_1,\LLL,\traceVar_n) 
    %\\&\ \ 
    \text{ is a multi-trace specification.} 
\end{cases}
\end{equation*}
\label{def:application_subst}
\end{definition}

%\begin{example}
%	Consider \ISO, as specified in the introduction,
%	$$
%	\obsQ{\forall}{\traceVar}\ \obsQ{\exists}{\traceVar'} \ (\lab(\traceVar')[0] = \LC \wedge \outVar(\traceVar) \mathop{=_{\LC}}\outVar(\traceVar')).
%	$$
%	and the substitution where \(\subst[\traceVar']\mathop{=}\noSecret(\traceVar)\) and \(\subst[\traceVar]\) is not defined.
%	%
%	Then, \(\varphi_{\text{\ISO}}[\subst]\) is:
%	$
%	\obsQ{\forall}{\traceVar}\ \genQ{\exists}{\traceVar'}{\noSecret(\traceVar)} \ (\lab(\traceVar')[0] = \LC \wedge \outVar(\traceVar) \mathop{=_{\LC}}\outVar(\traceVar')).
%	$\hfill $\blacktriangleleft$
%\end{example}

%\begin{definition}
%	\label{def:sound}
%	Let $\setTraces$ be a set of traces and $\varphi$ be a linear-time hyper formula over the set of trace variables $\VarTrace$.
	%
\vspace*{-4mm}
A substitution $\subst$ is \emph{sound for a set of traces $\setTraces$ and hyper formula $\varphi$} if whenever $\setTraces$ together with an interpretation for the generators satisfies the substituted formula ($(\setTraces,\funcSign)\mathop{\models} \varphi[\subst]$), it also satisfies the original formula ($(\setTraces,\funcSign)\models \varphi$).
%\end{definition}
%
The theorem below captures the standard soundness properties of substituted quantifiers: a substituted existential must act as an under-approximation that preserves witnesses, while a substituted universal must define an over-approximation that preserves obligations.

\begin{theorem}[Sound replacement of observing quantifiers]
\label{thm:sound_obs}\hfill\linebreak
Let
$\varphi \mathop{=} \obsQ{\Quant_1}{\traceVar_1} \LLL \obsQ{\Quant_n}{\traceVar_n}\ \psi(\traceVar_1, \LLL, \traceVar_n)$
where
$\psi(\traceVar_1,\LLL,\traceVar_n)$ is a multi-trace specification.
Let $\subst$ be a generator substitution such that, for every $\traceVar_i\mathop{\in}\{\traceVar_1,\LLL,\traceVar_n\}$ with $\subst[\traceVar_i]\mathop{=}\func(\overline{\traceVar})$ defined,
%\begin{align*}
$	\Quant_i\mathop{=}\exists \tAnd \func^{\funcSign}\mathop{\sqsubseteq}\setTraces,
\text{ or }
\Quant_i\mathop{=}\forall \tAnd \setTraces\mathop{\sqsubseteq}\func^{\funcSign}.$
%\end{align*}
If $(\setTraces,\funcSign) \models \varphi[\subst]$ then $\setTraces \models \varphi$.
\end{theorem}

\newcommand{\substS}{\subst_s}
\newcommand{\funcS}{\func_s}
\newcommand{\funcSignS}{\funcSign_s}

%\smallskip
\noindent\textbf{General Criteria for Sound Generator Substitutions.}
We now identify general semantic conditions under which generators can be substituted soundly, generalizing Theorem~\ref{thm:sound_obs} along two dimensions.
First, the criterion relates two generator interpretations (a reference $\funcSign$ and a substituted $\funcSignS$), so it also covers replacing one generator by another, not only replacing observed quantifiers by generators.
Second, the under- and over-approximation conditions are relaxed to a simulation condition between $\funcSign$ and $\funcSignS$ that has them as a special case, so the substituted generators need not approximate any observed set.
To simplify our presentation, in the main results below we consider only hyper formulas with generating quantifiers; the results extend naturally to observing quantifiers.
%
%
%From now on, \(\funcSign_s\) is a generators interpretation.

\begin{definition}
\label{def:form_correct}
%Let $\Func$ be a generators signature and $\VarTrace$ a set of trace variables.
Let
$
\varphi\mathop{=}\genQ{\Quant_1}{\traceVar_1}{\func_1(\overline{\traceVar}_1)} \LLL \genQ{\Quant_n}{\traceVar_n}{\func_n(\overline{\traceVar}_n)} \ \psi(\traceVar_1, \LLL, \traceVar_n)
$
be a hyper formula over a set of trace variables $\VarTrace$ where
$\psi(\traceVar_1,\LLL,\traceVar_n)$ is a multi-trace specification,
%$\func_i(\overline{\traceVar}_i)\mathop{\in}\mathbb{F}(\Func,\VarTrace)$, for $i \mathop{\leq} n$, 
and each generator is quantified at most once.
Let $\subst$ be a generator substitution.
The substitution $\subst$ and generator interpretation $\funcSignS$ are
\emph{sound for $\varphi$, a reference interpretation $\funcSign$ and $\setTraces$}, written
$
(\varphi[\subst],\funcSignS)\corrFI (\varphi,\funcSign),
$
iff $(\varphi[\subst],\traceAssign,\funcSignS)\corrFI(\varphi,\traceAssign,\funcSign)$ holds for every assignment $\traceAssign$,
where the assignment-indexed relation is defined recursively, as follows:

\smallskip\noindent
\emph{Base case: $\varphi \mathop{=} \psi$ is a multi-trace specification.}
$(\psi,\traceAssign_s,\funcSignS)\corrFI(\psi,\traceAssign,\funcSign)$ iff
$(\traceAssign_s,\setTraces,\funcSignS)$ $\models\psi \ \Longrightarrow\ (\traceAssign,\setTraces,\funcSign)\models\psi$.

\smallskip\noindent
\emph{Recursive case: $\varphi \mathop{=} \genQ{\Quant}{\traceVar}{\func(\overline{\traceVar})}\,\varphi'$.}
Let $\funcS(\overline{\traceVar}_s)\mathop{=}\subst[\traceVar]$ if $\subst[\traceVar]$ is defined and, otherwise,
$\funcS(\overline{\traceVar}_s)\mathop{=}\func(\overline{\traceVar})$, so 
$\varphi[\subst]\mathop{=}\genQ{\Quant}{\traceVar}{\funcS(\overline{\traceVar}_s)}\,\varphi'[\subst]$.
We say that two traces
$\trace_s$ and $\trace$ are \emph{compatible} when:
$
{\big(\varphi'[\subst],\, \traceAssign_s[\traceVar\mapsto(\trace_s,0)],\, \funcSignS\big)
\mathop{\corrFI}
\big(\varphi',\, \traceAssign[\traceVar\mapsto(\trace,0)],\, \funcSign\big)}.
$
Then, $(\varphi[\subst],\traceAssign_s,\funcSignS)\corrFI(\varphi,\traceAssign,\funcSign)$ holds iff:
\begin{itemize}
	\item When $\Quant=\exists$: every $\trace_s\mathop{\in} \funcS^{\funcSignS}\big(\traceAssign_s(\overline{\traceVar}_s)\big)$ is compatible with some $\trace\mathop{\in} \func^{\funcSign}\big(\traceAssign(\overline{\traceVar})\big)$;
	\item When $\Quant=\forall$: every $\trace\mathop{\in} \func^{\funcSign}\big(\traceAssign(\overline{\traceVar})\big)$ is compatible with some $\trace_s\mathop{\in} \funcS^{\funcSignS}\big(\traceAssign_s(\overline{\traceVar}_s)\big)$.
\end{itemize}
\end{definition}

Intuitively, the definition above requires a simulation-style relationship between the executions generated by $\funcSign$ and $\funcSign_s$.
For existential quantification, every execution generated under $\funcSign_s$ must be simulated by an execution generated under $\funcSign$ with respect to the continuation formula $\varphi'$.
For universal quantification, the direction of simulation is reversed.

\begin{theorem}[Soundness of Substitutions]
\label{thm:correct_ative}
Let
\(\varphi\) be a hyper formula where each generator appears in at most one generating quantifier.
Let \(\subst\) be a generator substitution.
Let $\setTraces$ be a set of traces,
\(\funcSign\) and \(\funcSignS\) be generator interpretations s.t.\
\((\varphi[\subst],\funcSign_s)\corrFI (\varphi,\funcSign)\). 
For all assignments \(\traceAssign\):
if \((\traceAssign, \setTraces,\funcSignS) \mathop{\models} \varphi[\subst]\)
then \({(\traceAssign, \setTraces,\funcSign)\mathop{\models} \varphi}\).
\end{theorem}

% ---- Monitoring Problem ----
\section{Monitoring Hyperproperties with Generators}
\label{sec:monitorability}
Before addressing monitorability for hyper formulas with generators, we revisit the main monitoring concepts for hyperproperties.
We denote the set of finite subsets of $\setTraces$ by
\(\boundPowerSet{}{\fin}{\setTraces}\mathop{=}\{U\mid U\mathop{\subseteq} T, {|U|  \mathop{\in} \nat}\}.\)
Likewise, \(\boundPowerSet{}{\term}{\setTraces}\mathop{=}\{U\mid U\mathop{\subseteq} T, \forall \trace\mathop{\in}U\ \trace[|\trace|\mathop{+}1]\mathop{=}\termSym\}\) collects the subsets containing only terminated traces\footnote{We recall that at runtime an observation may have unterminated traces, i.e., prefixes that can still be extended; the explicit marker~\(\termSym\) is what lets us tell the two apart.}.
Observations of a system are finite sets of finite traces, i.e., \(\observ\mathop{\in}\boundPowerSet{}{\fin}{\alphabet^*\{\termSym\}^*}\).
The set of observations that \emph{permanently satisfy or violate a given hyperproperty \(\hyper\) over \(\alphabet\)} is characterized by the following sets of good and bad prefixes~\cite{monitorHyper17}:
\vspace{-1mm}
\begin{align*}
	%----- Good prefixes
	\goodSet^{}(\hyper)\! = &\{\observ\mathop{\in} \boundPowerSet{}{\fin}{\alphabet^*\{\termSym\}^*} \, |\,
	\forall \mathcal{V} \mathop{\in} \boundPowerSet{}{\term}{\alphabet^*\{\termSym\}^*}\, \observ \hyperExtend  \mathcal{V} \mathop{\rightarrow} \mathcal{V} \mathop{\in} \hyper\}&&\\
	%----- Bad prefixes
	\badSet^{}(\hyper)\! = &\{\observ\mathop{\in} \boundPowerSet{}{\fin}{\alphabet^*\{\termSym\}^*} \, |\,
	\forall \mathcal{V} \mathop{\in} \boundPowerSet{}{\term}{\alphabet^*\{\termSym\}^*}\, \observ \hyperExtend  \mathcal{V} \mathop{\rightarrow} \mathcal{V} \mathop{\notin} \hyper\}.
\end{align*}

The \emph{monitoring problem} asks whether a finite observation suffices to decide a hyperproperty.
Formally, given a hyperproperty $\hyper$ over $\alphabet$ and an observation \(\observ\mathop{\in}\boundPowerSet{}{\fin}{\alphabet^*\{\termSym\}^*}\),  it asks whether \(\observ\) permanently satisfies or violates~\(\hyper\).
%
%A property is \emph{monitorable} when every observation can be extended to one solving this problem.

\begin{definition}
	\label{def:mon}
	A hyperproperty \(\hyper\) over \(\alphabet\) is \emph{monitorable} iff every observation
	\(\observ \mathop{\in} \boundPowerSet{}{\fin}{\alphabet^*\{\termSym\}^*}\) has an extension \(\observ' \mathop{\in} \boundPowerSet{}{\fin}{\alphabet^*\{\termSym\}^*}\), \(\observ \hyperExtend \observ'\), that lies in the good or bad set of \(\hyper\), i.e., \(\observ' \mathop{\in}\goodSet^{}(\hyper) \mathop{\cup}\badSet^{}(\hyper)\).
\end{definition}

\begin{example}
	Consider the following observed traces with the same public inputs but different outputs:
	$\trace_0\mathop{=} \{\pvF{\inlow}\!\As \text{\ttfamily 1}, \pvF{\inhigh}\!\As \text{\ttfamily 1},$ $\pvF{\outlow}\!\As \text{\ttfamily 1}\}$
  and
	$\trace_1=\{\pvF{\inlow}\!\As \text{\ttfamily 1}, \pvF{\inhigh}\!\As \text{\ttfamily 0},$ $\pvF{\outlow}\!\As \text{\ttfamily 0}\}.$
	These two traces witness a violation of OD.
	Moreover, any observation \(\observ\) extended with them, \(\observ'\mathop{=}\observ\mathop{\cup}\{\trace_0,\trace_1\}\), is a bad prefix for OD (i.e., it witnesses a permanent violation).
	\hfill $\blacktriangleleft$
\end{example}

Importantly, the monitoring problem is defined exclusively over observed executions.
Consequently, the framework remains unchanged when the monitored hyperproperty includes generating quantifiers, since these quantifiers belong to the specification and do not affect the space of observations.

\subsection{Monitorability of Linear-time Hyper Specifications}
We characterize monitorability for linear-time hyper specifications whose observing quantifiers form a uniform, alternation-free prefix, with unrestricted generating quantifiers after it.
%; we call \(\varphi\) and \(\funcSign\) monitorable when \(\hyperSet{\varphi,\funcSign}\) is.
%
Our results rely on monotonicity of the satisfaction relation for specifications with only generating quantifiers.

\begin{lemma}
\label{prop:monotonic_hyper_gen}
Let \(\varphi\)
be a hyper specification with only generating quantifiers
and \(\funcSign\) an interpretation for its generators.
For all sets of traces \(\setTraces\) and \(\setTraces'\) over \(\alphabet\) s.t.\ \(\setTraces\mathop{\subseteq}\setTraces'\):
if \((\setTraces,\funcSign)\mathop{\models}\varphi\)
then
\((\setTraces',\funcSign)\mathop{\models}\varphi\).
\end{lemma}

Formulas with a leading prefix of observing universal quantifiers either have no good prefixes or hold for all models; dually for existential quantifiers.

\begin{proposition}
\label{prop:univ:goodPre}
Let \(\psi\) be a hyper specification with only generating quantifiers and \(\funcSign\) an interpretation for its generators.
For \(\varphi_{\forall}\mathop{=}\obsQ{\forall}{\traceVar_1}\LLL\obsQ{\forall}{\traceVar_n}\ \psi\),
either \(\goodSet(\hyperSet{\varphi_{\forall},\funcSign})\mathop{=}\emptyset\) or
\(\hyperSet{\varphi_{\forall},\funcSign}\mathop{=}\boundPowerSet{}{\term}{\alphabet^*\{\termSym\}^*}\).
For \(\varphi_{\exists}\mathop{=}\obsQ{\exists}{\traceVar_1}\LLL\obsQ{\exists}{\traceVar_n}\ \psi\),
either \(\badSet(\hyperSet{\varphi_{\exists},\funcSign})\mathop{=}\emptyset\) or
\(\hyperSet{\varphi_{\exists},\funcSign}\mathop{=}\emptyset\).
\end{proposition}

Specifications with only universal observing quantifiers either hold for all models or are monitorable only for violation; dually for existential quantifiers.

\begin{theorem}[Monitorability of Hyper Specifications]
\label{thm:mon}
Let \(\psi\) be a linear-time hyper specification with only generating quantifiers s.t.\ \(\Free{\psi}\mathop{=}\{\traceVar_1, \ldots, \traceVar_n\}\), and \(\funcSign\) an interpretation for its generators.
The hyper formula \(\varphi_{\forall}\mathop{=}\obsQ{\forall}{\traceVar_1} \LLL \obsQ{\forall}{\traceVar_n}\ \psi\) and \(\funcSign\)
are monitorable iff either all systems and assignments are models of \(\psi\) under \(\funcSign\)
(i.e., for all \(\traceAssign\) and \(\setTraces\): \((\traceAssign, {\setTraces},\funcSign)\mathop{\models}\psi\))
or \(\badSet(\hyperSet{\varphi_{\forall},\funcSign})\mathop{\neq}\emptyset\).
The hyper formula \(\varphi_{\exists}\mathop{=}{\obsQ{\exists}{\traceVar_1} \LLL \obsQ{\exists}{\traceVar_n}\, \psi}\) and \(\funcSign\)
are monitorable iff \(\psi\) under \(\funcSign\) has no models
(i.e., \(\hyperSet{\varphi_{\exists}, \funcSign} \mathop{=} \emptyset\))
or \(\goodSet(\hyperSet{\varphi_{\exists},\funcSign})\mathop{\neq}\emptyset\).
\end{theorem}

In general, specifications with alternating observing quantifiers are not
monitorable~\cite{monitoringHyperLTL19,stucki2021gray}.
Theorem~\ref{thm:mon} characterizes monitorability for specifications whose observing quantifiers form a uniform, alternation-free prefix followed by arbitrary generating quantifiers.
Together with the results on sound substitutions (Theorems~\ref{thm:sound_obs} and~\ref{thm:correct_ative}), this lets us soundly map alternating specifications into the fragment characterized by Theorem~\ref{thm:mon}.

\vspace*{-2mm}
\begin{example}
\label{ex:iso_monitorable}
Looking back at our \ISO example: it is, in general, not monitorable, but replacing the inner observing quantifier % of its violation form
by the generator $\noSecret$ (as done in~\eqref{def:intro:opacity:gen}) yields a formula with a uniform existential observing prefix, placing it in the fragment of Theorem~\ref{thm:mon}.
This substitution is sound by Theorem~\ref{thm:sound_obs} whenever $\noSecret$ over-approximates the system.
This may, however, not be practical, because an
over-approximation requires every system execution to appear among $\noSecret$'s outputs, forcing $\noSecret$ to emit complete executions.
The simulation criterion of Theorem~\ref{thm:correct_ative} relaxes this: it requires only that each execution be matched by an output agreeing on what the body inspects (here, the public outputs and initial label) and never on an execution's hidden behavior, so $\noSecret$ may fix those observations and abstract the rest while staying sound for monitoring violations.
\hfill $\blacktriangleleft$
\end{example}

\vspace*{-6mm}
\subsection{Monitoring Algorithm for Hyperproperties with Generators} \label{sec:monitoring}
The monitoring algorithm with generator functions builds on monitoring algorithms for a chosen linear-time hyper specification: the only requirements are that we can build monitors for universally quantified prenex formulas of this logic, we can form dual formulas in the logic ($\exists \pi\ \phi \equiv \neg (\forall \pi\ \neg \phi)$), and that we can instruct the monitor to assign concrete traces to given variables.
We call these monitors \emph{basic monitors}.

Our monitoring algorithm is a combinatorial algorithm~\cite{hahn19} that, in the worst case, explicitly examines every instantiation of quantifiers.
Briefly, the algorithm works by instantiating quantifiers outer- to inner-most (from left to right) until we are left with a universally quantified formula with quantifiers ranging over the same set of traces; this formula is then handled by a basic monitor. The instantiation itself is done through sub-monitors, each taking care of the ``next'' part of the quantifier prefix.
Existential quantifiers are turned into universal ones during instantiation by duality ($\exists \pi\ \phi \equiv \neg (\forall \pi\ \neg \phi)$) and sub-monitors take care of translating back results after negation.
%The process of quantifier instantiation is analogous to handling quantifier alternation in automata-based LTL model checking~\textcolor{red}{\ref{ltl-modelchecking}}.
To concretize the intuition, consider the formula
$\obsQ{\forall}{\traceVar}\, {\genQ{\exists}{\traceVar'}{\noSecret(\traceVar)}} \, (\outlow(\traceVar) \mathop{\neq}\outlow(\traceVar'))$.
%\begin{flalign*}
%\obsQ{\forall}{\traceVar}\ \genQ{\exists}{\traceVar'}{\noSecret(\traceVar)} \ (\pv{label}{\traceVar'}[0] = \text{public} \wedge \pv{pub}{\traceVar}=\pv{pub}{\traceVar'}).
%\end{flalign*}
The top-level sub-monitor instantiates $\traceVar$ with every observed trace $\trace$ and delegates the remaining formula $\genQ{\exists}{\traceVar'}{\noSecret(\trace)}\ (\outlow(\trace)\not=\outlow(\traceVar'))$ to a child sub-monitor.
Because this formula is existentially quantified, the child sub-monitor hands its dual $\genQ{\forall}{\traceVar'}{\noSecret(\trace)}\ (\outlow(\trace)\mathop{=}\outlow(\traceVar'))$ to a basic monitor --- which it can do, as the dual is universally quantified and ranges over a single set of traces --- and negates the result it gets back.

Before giving a detailed description of the algorithm, we need to introduce a few preliminaries.
We say that two quantifiers are of the \emph{same type} if they range over the same domain (observed traces or the same set of generated traces) and they are both either universal, or both existential.
%In other words, universal and existential quantifiers have different types, but also quantifiers ranging over different traces sets have different types.
For a quantifier $Q$, we denote flipping it from existential to universal and vice versa by $\overline{Q}$.

\newcommand{\subs}{\va{submonitors}}

\SetKwFunction{step}{\textnormal{\textsc{step}}}
\SetKwFunction{ehlmon}{\textnormal{\textsc{basic-mon}}}
\SetKwFunction{stepehl}{\textnormal{\textsc{step-basic}}}
\SetKwFunction{handleehl}{\textnormal{\textsc{handle-basic}}}

\SetKwFunction{submon}{\textnormal{\textsc{SubMon}}}
\SetKwFunction{basicmon}{\textnormal{\textsc{basic-mon}}}

\SetKwInput{KwReq}{require}

\IncMargin{1em}

\begin{algorithm}[t]
  \DontPrintSemicolon
  \small
  \KwData{Closed multi-trace prenex formula $\psi = Q_1^* Q_2^*...Q_n^*\ \psi_{qf}$ with $n > 0$}
  \KwData{A set $S$ of observed traces, generator interpretation $\sigma$}
  \KwResult{The truth value of $\psi$ for given observations and generators.}
  \vspace*{.8em}
  
  %\SetKwFunction{step}{\textnormal{\textsc{step}}}
  \SetKwProg{Proc}{Procedure}{}{}

  \Proc{\submon{$F, \Pi$}} {
    let $Q_{j_1}^*Q_{j_2}^*...Q_{j_k}^*\ \psi_\mathit{qf}' = F$\; 
    \BlankLine
    \eIf{$Q_{j_1}^*$ ranges over $f(\pi_1, ...)$}{
      $T \gets f^\sigma(\Pi(\pi_1), ...)$ \tcp{Traces generated by $f$}
    }{
      $T \gets S$ \tcp{Observed traces}
    }
    \BlankLine
    \If {$k = 1$} { \label{alg:step_mon:ln-k1} \label{alg:mon_gen:ln-basic}
      \eIf{$Q_{j_1}^*$ are existential}{
        \KwRet $\neg$\basicmon{$\overline{F}, \Pi, T$}
      }{
        \KwRet \basicmon{$F, \Pi, T$}
      }
    }
    \BlankLine
      
    \eIf {$Q_{j_1}^*$ are existential} {
      $F' \gets \overline{Q}_{j_2}^*...\overline{Q}_{j_k}^*\ \neg\psi_\mathit{qf}'$\label{alg:mon_gen:ln-neg1}\;
    }{
      $F' \gets Q_{j_2}^*...Q_{j_k}^*\ \psi_\mathit{qf}'$\;
    }

    \ForEach{trace assignment $\Pi'$ of $Q_{j_1}^*$ to $T$} { \label{alg:mon_gen:ln-inst}
      $r \gets$ \submon{$F', \Pi \cup \Pi'$}\label{alg:mon_gen:ln-rec}\;
      \tcp{Did we find an unsatisfying assignment?}
     \lIf {$Q_{j_1}^*$ are universal and $r = \mathit{false}$} {
       \KwRet $\mathit{false}$
     }\label{alg:mon_gen:ln-univ}
      \lIf {$Q_{j_1}^*$ are existential and $r = \mathit{false}$} {
        \KwRet $\mathit{true}$
      }\label{alg:mon_gen:ln-neg2}
    }
    \tcp{No unsatisfying assignment found}
    \leIf {$Q_{j_1}^*$ are existential} {\KwRet $\mathit{false}$}{\KwRet $\mathit{true}$}
  }
  \BlankLine

  \tcp{The top-level monitor}
  \KwRet \submon{$\psi, \emptyset$}\label{alg:mon_gen:ln-top}
\caption{The monitoring algorithm with generator functions.}
\label{alg:mon_gen}
\end{algorithm}

The algorithm, shown in Algorithm~\ref{alg:mon_gen}, groups consecutive same-type quantifiers into \emph{blocks} and processes one block of quantifiers in each recursive call of $\submon$.
Sub-monitors form a tree: each sub-monitor spawns one child sub-monitor per trace assignment to its block of quantifiers $Q_{j_1}^*$ (lines~\ref{alg:mon_gen:ln-inst} and \ref{alg:mon_gen:ln-rec}), and the leaves are basic monitors (line~\ref{alg:mon_gen:ln-basic}). The root of the tree is the sub-monitor applied to the whole input formula %, handling the first block of quantifiers~
(line~\ref{alg:mon_gen:ln-top}).
Quantifiers are instantiated either from observed traces or from generated sets of traces.
Notice that whenever trace variables should be instantiated by a generator $\func(\traceVar_1,\LLL, \traceVar_m)$, the variables $\traceVar_1,\LLL, \traceVar_m$ have already been instantiated as they must come from quantifiers ``to the left'', so we can always obtain the set of generated traces.

Existential quantifiers are handled by duality during recursion. When the current block of quantifiers is existential, the sub-monitor negates the remaining formula (line~\ref{alg:mon_gen:ln-neg1}) and lets child sub-monitors evaluate it as a universal formula.
After the child monitor returns, the sub-monitor negates the result back (line~\ref{alg:mon_gen:ln-neg2}).
This way every sub-monitor internally handles only universal quantifiers.

Note that recursive calls build trace assignments without actually substituting the assignment into the formula itself; the assignment is kept separately from the formula.
This is on purpose, because such a formula may not even be constructible in the given logic and if it is, the formula would become very large.
When the formula is reduced to a formula with a single block of quantifiers, a basic monitor is invoked together with the computed quantifier assignment (line~\ref{alg:mon_gen:ln-basic}), leaving it on the basic monitor to handle the assigned traces (which it can do per our assumption on basic monitors).

\vspace*{-2mm}
\paragraph{Correctness (sketch)}
By induction on $k$. For $k=1$, the correctness follows from the assumption that basic monitors are correct.
For $k > 1$, assume
that child sub-monitors correctly evaluate their formula to $\va{true}$ or $\va{false}$.
Then for a universally quantified formula, the monitor returns $\va{false}$ iff some child returns $\va{false}$ (line~\ref{alg:mon_gen:ln-univ}), which means that the formula is not satisfied and the answer is correct.
For an existentially quantified formula, the monitor returns $\mathit{true}$ iff some sub-monitor returns $\mathit{false}$ for the dual formula, (i.e., a witness was found), and returns $\va{false}$ iff all sub-monitors return $\va{true}$ for the dual formula. %(i.e., no witness was found).

%\paragraph{Discussion}
% Our algorithm is in essence a naive evaluation of the input formula on all possible combinations of traces.
% Its usefulness comes precisely from the generator functions.
% As discussed in the previous sections, if we use this algorithm only with observed traces and the monitored formula contains a quantifier alternation, the obtained result is inconclusive.
% With generator functions, we can get a meaningful result.
\smallskip
The assumption that basic monitors must handle prenex universally quantified formulas is not restrictive, because this is the setup of most known approaches for monitoring hyperproperties~\cite{Agrawal16,brett17,monitorHyper17,monitoringHyperLTL19,hahn19,Hahn19rv,AcetoAAF22,prefixTransducers23}.
Similarly, the duality principle holds in commonly used logics.
We note that we use duality mainly to approach sub-monitors uniformly:
the quantifier instantiation %part of the algorithm
could be formulated without doing negations and the only place where negation would be really needed is to obtain a universally quantified formula for the basic monitor in the last recursion step (unless basic monitors handle existentially quantified formulas, in which case we could remove both of the assumptions discussed above).

For ease of presentation, we described the algorithm as an offline algorithm.
The algorithm can be made online by evaluating all sub-monitors lazily on-demand as traces and events come.
A detailed description of the online algorithm is provided in the extended version of this paper~\cite{arxiv-ext}.

\section{Empirical Evaluation}
\label{sec:evaluation}
In this section, we evaluate the feasibility of monitoring hyperproperties with generators, and study how they affect violation detection, quantifier handling, and monitoring efficiency.
%
%In particular, our benchmarks are designed to answer the following questions: 
%
% \anaSide{We need to answer these questions.}
% \begin{enumerate}
% \item[(Q1)] \emph{To what extent do generators accelerate the detection of violations?}
% %
% \item[(Q2)] \emph{Can generators enable monitoring of $\forall \exists$ hyperproperties that are otherwise not monitorable using only observed traces?}
% %
% \item[(Q3)] \emph{How does generator precision affect the trade-off between soundness, accuracy, and monitoring efficiency?}
% %
% \item[(Q4)] \emph{Can generators be used to model idealized executions sufficiently well for monitoring?}
% \end{enumerate}
% 
We have implemented the online version of our algorithm with basic monitors for \shl~\cite{hypernodeMonInfinite25} and $\HLTL$, demonstrating the applicability of our approach across different hyperlogics.
The implementation is a Python project that generates monitors in C++~\cite{rvhyno}.

We evaluate the implementation on monitoring OD, \ISO, linearizability, and generalized non-interference (GNI).
The experiments were run on a laptop with an 11th Gen Intel(R) i7 @ 2.8 GHz processor and 16 GB of RAM.
The timeout for each monitor run was $30$\,s, except for the GNI experiments, where it was $60$\,s.

%We have implemented our monitoring algorithm with basic monitors for \emph{Symbolic hypernode logic (\shl)}~\cite{hypernodeMonInfinite25}
%and partially for \emph{HyperLTL} and evaluated them on three use cases:
%  1) using an under-approximative generator to speed up violation detection when monitoring observational determinism (OD), which is a universally quantified formula (Section~\ref{ssec:eval:od}),
%  2) using an over-approximative generator function to monitor initial-state opacity (ISO), which is a $\forall\exists$ hyperproperty (Section~\ref{ssec:eval:op}), and
%    3) using a generator to obtain linearized traces to monitor linearizability and a sequential property of concurrent traces (Section~\ref{ssec:eval:lin}).
% \begin{itemize}
%   \item using an under-approximative generator function to speed up violation detection when monitoring \emph{observational determinism (OD)}, which is a universally quantified formula (Section~\ref{ssec:eval:od}),
%   \item using an over-approximative generator function to monitor \emph{initial-state opacity (ISO)}, which is a $\forall\exists$ hyperproperty (Section~\ref{ssec:eval:op}), and
%     \item using a generator function to obtain linearized traces to monitor linearizability and a sequential property of concurrent traces (Section~\ref{ssec:eval:lin}).
% \end{itemize}

%\vspace{-4mm}
\smallskip
\noindent\textbf{Observational Determinism.} %~~\label{ssec:eval:od}
In this experiment, we model a robot moving on a $10\times 10$ grid which is
partitioned into areas (groups of adjacent cells).
The robot takes a target cell as a secret input, and a publicly observable sequence of areas that it must go through on its way from the initial cell to the target cell (the robot can also visit other areas). % between any two consecutive input areas).
The state of the robot, i.e., in which cell it is, is hidden at all times from the observer (e.g., an attacker), who can observe only in which area the robot is at every step.

We let the robot use a randomly generated deterministic strategy,
and we monitor whether the strategy leaks information about the hidden target by monitoring the formula
\begin{equation}\label{eq:eval_od}
   \obsQ{\forall}{\traceVar}\obsQ{\forall}{\traceVar'}\ \ (\pv{input}{\traceVar} \mathop{=} \pv{input}{\traceVar'} \rightarrow \pv{area}{\traceVar} \mathop{=} \pv{area}{\traceVar'})
   %\forall \traceVar \forall \traceVar'\ \ \pv{input}{\traceVar} \le \pv{input}{\traceVar'} \rightarrow \pv{area}{\traceVar} = \pv{area}{\traceVar'}
\end{equation}
and the formula with a generator $\mathtt{samples}$
\begin{equation}\label{eq:eval_od_fun}
\obsQ{\forall}{\traceVar}\, 
    {\genQ{\forall}{\traceVar'\!\!}{\!\mathtt{samples}(\traceVar)}} \ \pv{area}{\traceVar} \mathop{=} \pv{area}{\traceVar'}
\end{equation}
where $\mathtt{samples}$ returns randomly sampled traces with inputs from the trace $\traceVar$.
%\newcommand{\bullet}[2][red,fill=red]{\tikz[baseline=-0.5ex]\draw[#1,radius=#2] (0,0) circle ;}%

% Notice that in Equation~(\ref{eq:eval_od}), we check only for prefixing of inputs instead of for equality.
% The reason for this is that inputs are given at the beginning of the trace and then the rest of the inputs projection consists of a special delimiter symbol and a padding with default symbols.
% Two traces may have the same inputs but different projection to $\pvF{input}$, because the padding with the default symbol can have different lengths (if the traces have different lengths).
% For example, we can have traces with the input projection $AB\bullet\#\#$ and $AB\bullet\#$ (where $\bullet$ is the delimiter symbol and $\#$ the padding symbol).
% Clearly, the traces have same inputs, but the sequence of symbols is not equal.
% That is why we check only for $\pv{input}{\traceVar} \le \pv{input}{\traceVar'}$ in Equation (\ref{eq:eval_od}).
% Because OD is a symmetric property, this is sufficient for determining all pairs of traces that need to be checked for OD. %for equal sequences of areas.
%
\begin{figure*}[t]
  \centering
\includegraphics[width=\textwidth]{
	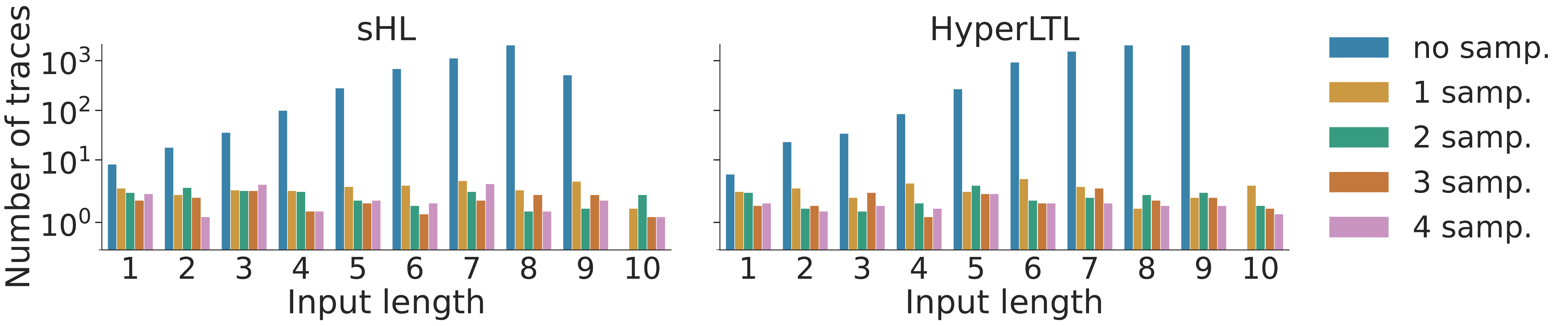
}
%\vspace{-3mm}
\caption{The number of traces required to find an OD violation by a monitor without (\textit{no samp.}), and with a generator (\textit{$n$ samp.}, where $n$ is the number of samples the function returns). The left sub-plot shows results for \shl monitors while the right subplot for $\HLTL$ monitors.}
\label{fig:traces_to_violation}
\vspace*{-4mm}
\end{figure*}

Fig.~\ref{fig:traces_to_violation} compares monitoring with and without the generator, showing how many traces were needed for detecting the OD violation.
We considered two types of basic monitors, one specified in \shl and the other in $\HLTL$ (the variables were broken down into propositions for $\HLTL$).
The numbers are averages of 10 repetitions of the experiment (with different traces each trial).
With the generator function, the violation was typically detected after one or two traces were observed.
Without it, the detection could take hundreds of traces and even fail --- a violation went undetected in many cases because the set of observations did not contain traces that would witness the violation (monitors were set to give up after 2048 traces).
If the input length exceeded 9 areas, the probability of observing a pair of violating traces decreased so much that the monitor without generator never detected a violation.
When a violation was not detected, we had no number of traces to report, which explains the missing blue bar for inputs of length 10 and the drop in the blue bar for inputs of length 9 in the left plot.

\smallskip
\noindent\textbf{Initial-state Opacity.} %~~\label{ssec:eval:op} 
In these experiments, we again simulate a robot navigating a $10\times 10$ grid that is
partitioned into publicly observable areas, while concrete states and actions of the robot remain hidden from an observer.
The robot travels from one of several secret initial locations to a target area according to a non-deterministic strategy, passing through a sequence of user-defined checkpoints.
%, that is, the observer can only see in which area the robot is at each step.
%and 
The goal is to monitor \ISO: that an outside observer cannot determine the initial location from the observations.
In \shl, we can specify this property with the formula
%
%\vspace*{-4mm}
\begin{equation}\label{eq:eval_op}
%$
\obsQ{\forall}{\traceVar}
\obsQ{\exists}{\traceVar'}\  (\pv{area}{\traceVar} \mathop{=} \pv{area}{\traceVar'} \land \pv{act}{\traceVar} \neq \pv{act}{\traceVar'})
\end{equation}
where $\pvF{act}$ is the sequence of move actions performed by the robot. %; the actions can be (move) \emph{left}, \emph{right}, \emph{up}, \emph{down}, or the action \emph{stand}.

We define a generator function $\mathtt{eqarea}(\trace) \mathop{=} \{\trace' \mid \pv{area}{\trace'}\mathop{=}\pv{area}{\trace}\}$ returning for a trace $\trace$ all traces with the same public observations (but possibly different actions).
We also define an alternative version of this function that directly returns a single trace $\trace_w$ for which it holds that $\pv{area}{\trace} \mathop{=} \pv{area}{\trace_w} \land \pv{act}{\trace} \mathop{\neq} \pv{act}{\trace_w}$,
i.e., this other version is tailored to checking opacity.
We refer to monitors using these functions as \emph{AAT} (all admissible traces) and \emph{1W} (one witness), resp.
The formula for monitoring opacity then becomes
\begin{equation}\label{eq:eval_op_fun}
%$
\obsQ{\forall}{\traceVar} \genQ{\exists}{\traceVar'}{\mathtt{eqarea}(\traceVar)}\  (\pv{act}{\traceVar} \mathop{\neq} \pv{act}{\traceVar'}).
%$
\end{equation}

Plots in Fig.~\ref{fig:opacity}
show the results of monitoring \ISO with these two generators and \shl basic monitors.
%The plot on the left shows the CPU time of monitoring for input length equal 8.
The left plot reveals that \emph{AAT} performs much worse than \emph{1W}, highlighting that tailoring functions %to fit
for given problems can have a huge impact.
The right plot shows the time of running the \emph{1W} monitor for varying lengths of input and number of observations.
As inputs grow longer, fewer traces can be processed, because it becomes harder for $\mathtt{eqarea}$ to find the witness.

\begin{figure}[t]
%\vspace{-5mm}
\centering
  \begin{tabular}{c p{4mm} c}
\includegraphics[width=.4\textwidth]{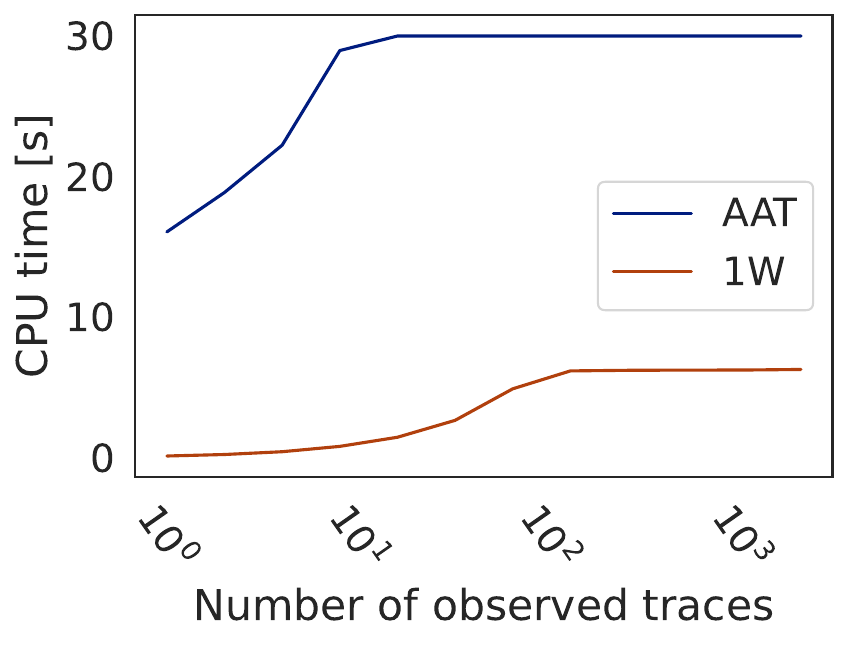}
&&
%\hspace{-5mm}
\includegraphics[width=.4\textwidth]{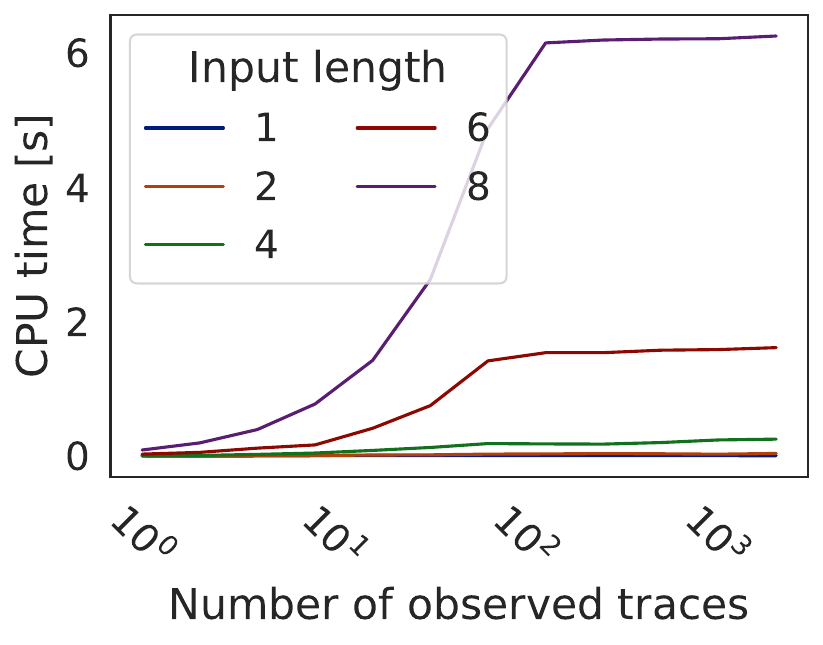}
  \end{tabular}
  %\vspace{-5mm}
  \caption{The left plot shows the CPU time of monitoring \ISO with \emph{1W} and \emph{AAT} monitors for input length equal to 8.
  The right plot shows the CPU time of monitoring \ISO with the \emph{1W} monitor.
  }
\label{fig:opacity}
%\vspace{-6mm}
\vspace*{-4mm}
\end{figure}

%\vspace*{-6mm}
\smallskip
\noindent\textbf{Linearizability.} %~~\label{ssec:eval:lin}
The next experiment is monitoring linearizability and properties on linearized traces.
%
%Linearizability \cite{LinearHerlihyWing90} is a consistency property for concurrent data structures, requiring all histories of invocations and response events to shared resources to be consistent with some sequential execution of those operations. 
%
%While verifying for \LIN
%, it is irrelevant whether the implementation of a concurrent data structure exhibits linear histories because 
%all that matters is that the implementation histories can be successfully related to the ``idealized'' sequential implementation.
%
We obtained concurrent traces from the Michael-Scott queue~\cite{msqueue} implementation in the \emph{Scal} benchmark suite~\cite{HaasHKLPS15}.
%The traces are operation-based, that is, instead of containing events with invoke and response, the events are operations (push and pop) with invoke and response time. This slightly simplifies the situation as we do not need to search for matching events. 
For each trace, we checked the following \shl formulas:
\begin{itemize}
\item[1)] \(\obsQ{\forall}{\pi} \genQ{\exists}{\pi_l}{\mathtt{lin}(\pi)}\  \mathit{true}\), 
\item[2)] \(\obsQ{\forall}{\pi} \genQ{\exists}{\pi_l}{\mathtt{lin}(\pi)} \genQ{\exists} {\pi_c}{\mathtt{legal}(\pi_l)}\  \mathit{true}\), and 
\item[3)] \(\obsQ{\forall}{\pi}\genQ{\exists}{\pi_l}{\mathtt{lin}(\pi)}\  \mathit{bounded}(\pi_l, 2)\).
\end{itemize}
The generator $\mathtt{lin}$ implements a brute-force backtracking algorithm that returns a singleton set with a linearization of its input trace or the empty set if the trace is not linearizable.
%\footnote{
%This is the current bottle-neck of our monitors, but the intention of these experiments is to show feasibility, not to compete with the state-of-the-art monitors.}.% but because we want to only show the feasibility of monitoring linearizability with our approach, and not to compete with specialized algorithms for this task, we decided to use this simple implementation.
%Because functions brings modularity into \shl, we can easily replace the implementation of $\mathit{lin}$ with one of the efficient algorithms for queue linearizability in the future.
%
Generator $\mathtt{legal}$ re-executes its input trace on the  \emph{C++ STL} queue implementation,
validating that the linearized trace is indeed legal
%Using this function amounts to validating that the trace returned by $\mathtt{lin}$ is legal and can be actually executed on an independent queue implementation.
%Using this function amounts to validating that the trace returned by $\mathtt{lin}$ is 
and can actually be executed on an independent queue implementation.
This is a form of certificate checking.
Finally, the formula 3) uses the predicate $\textit{bounded}(\pi_l, 2)$ to check whether the size of the queue never exceeds 2 (this predicate can be defined using a regular expression in \shl).
% (this regular expression, however, blows up in the size, so we do not show it here for the sake of brevity -- in our implementation, we specified this predicate using an automaton instead of a regular expression, because the automaton is small and simple). %\xxx{show the automaton}).
%Note that
It is necessary to check this property on linearized traces, because
%on concurrent traces, the monitor could report false violations.
otherwise the monitor could report false violations.

% Because the active specification is basically a trace property, we ranit makes more sense to run a monitor on every trace instead of one monitor on all traces. This way we get better data. The difference is that if we give all the traces to the monitor and there is a trace on which the monitor times out, then the rest of the traces will not be analyzed, while if we input each trace to its monitor, we avoid this problem. Given that there would be no timeouts, both ways would yield the same results.}
%
We ran monitors on traces in which 2 (3) producers push \emph{nops} numbers into the queue, and 2 (3, resp.) consumers pop them, all in parallel.
Therefore, a trace has the length $4\cdot\mathit{nops}$ ($6\cdot\mathit{nops}$, resp.).
The results summarized in Tab.~\ref{tab:lin} demonstrate that it is feasible to use our approach for monitoring linearizability and monitoring sequential properties over concurrent traces.
%The time spent in $\emph{lin}$ is clearly the substantial part of the runtime of the monitors.

\begin{table}[t]
  \centering
\caption{The CPU time of monitoring linearizability (\emph{lin}), linearizability with a check for the legality of the witness trace (\emph{lin + legal}),
  	and boundedness of the queue on the linearized trace (\emph{lin + bounded}).
  	The CPU time is in seconds.}
  \label{tab:lin}
\begingroup
\setlength{\tabcolsep}{8pt}

\begin{tabular}{lrrrrrrr}
\toprule
			    & \multicolumn{4}{c}{2 prod. + 2 cons}
			    & \multicolumn{3}{c}{3 prod. + 3 cons}\\
\cmidrule(lr){2-5}
\cmidrule(lr){6-8}
	Monitor/\emph{nops} & 200 & 400 & 600 & 800 & 200 & 400 & 600\\
\midrule
\textit{Lin.} & 1.97  & 7.85  & 25.20 & -          & 6.83  & 24.62 & -\\
\textit{Lin. + legal} & 1.97  & 7.83  & 25.14 & -  & 6.82  & 24.59 & -\\
\textit{Lin. + bounded} & 1.96  & 7.80  & 25.05 & -& 6.81  & 24.55 & -\\
\bottomrule
\end{tabular}
\endgroup
\end{table}

%\vspace*{-4mm}
\smallskip
\noindent\textbf{Generalized Non-Interference.} %~~\label{ssec:eval:gni}
The last experiment is monitoring generalized non-interference~\cite{mclean1996general}, a security hyperproperty that we can specify with the following \shl formula:
\begin{multline}
	\obsQ{\forall}{\traceVar_1, \traceVar_2}\ \obsQ{\exists}{\traceVar_3}\\ \
    \inlow(\traceVar_1) \mathop{=} \inlow(\traceVar_3) \land 
    \outlow(\traceVar_1) \mathop{=} \outlow(\traceVar_3)
    \land\
    \inhigh(\traceVar_2) \mathop{=}\inhigh(\traceVar_3)
	\label{def:eval:gni}
\end{multline}
In words, for any two traces $\traceVar_1$ and $\traceVar_2$, there must exist a third trace $\traceVar_3$ with the same publicly observable behavior as $\traceVar_1$ while having the same secret inputs as $\traceVar_2$.
We implemented a monitor for (\ref{def:eval:gni}) and a monitor for the formula
\begin{align}
	\obsQ{\forall}{\traceVar_1, \traceVar_2}\ 
     \genQ{\exists}{\traceVar_3}{\texttt{se\_find}(\traceVar_1, \traceVar_2)}
     \ \mathit{true}
	\label{def:eval:gni-se}
\end{align}
where the existential quantifier is instantiated from a set of traces generated by the function $\texttt{se\_find}$.
This function performs a targeted search for $\traceVar_3$ using \emph{symbolic execution (SE)}~\cite{King76} of the system and was implemented by calling a modified version of the symbolic executor \emph{BubaaK-LEE}~\cite{bubaak}.
The SE searches for traces fully satisfying the quantifier-free part of formula~(\ref{def:eval:gni}), therefore, the body of formula~(\ref{def:eval:gni-se}) has nothing left to check and can be $\mathit{true}$: the sole existence of a trace found by symbolic execution is enough.
For evaluation, we created 30 benchmarks in C or C++, some taken from the \emph{Securibench Micro} benchmarks~\cite{securibench} and some encoding known software vulnerabilities.
Details about the benchmarks can be found in the extended version of this paper~\cite{arxiv-ext}.

% Needs \usepackage{booktabs,graphicx,multirow,xcolor} and amssymb (for \checkmark).
\providecolor{tabsec}{HTML}{1B7F3B}
\providecolor{tabvio}{HTML}{C0392B}
\providecommand{\rot}[1]{\rotatebox[origin=l]{90}{#1}}
\providecommand{\rotc}[1]{\rotatebox[origin=c]{90}{#1}}
\providecommand{\vsec}{\textcolor{tabsec}{\ensuremath{\checkmark}}}
\providecommand{\vvio}{\textcolor{tabvio}{\ensuremath{\times}}}
\providecommand{\vna}{--}
\begin{table}[t]
\centering
\caption{%
  The results of monitoring GNI. The verdict \textcolor{tabsec}{$\checkmark$}~means that the monitor said GNI holds, while \textcolor{tabvio}{$\times$}~marks a detected violation. Times are wall time.
Rows \emph{gen.} (\emph{no gen.}) state the results with the generator (without it, resp.).
}
\label{tab:gni-results}
\setlength{\tabcolsep}{3pt}
\small
\resizebox{\textwidth}{!}{%
  \begin{tabular}{p{.5mm}ll*{15}{c}}
\toprule
 & & & \texttt{B1} & \texttt{B2} & \texttt{B3} & \texttt{B4} & \texttt{B5} & \texttt{B6} & \texttt{B7} & \texttt{B8} & \texttt{B9} & \texttt{B10} & \texttt{B11} & \texttt{B12} & \texttt{B13} & \texttt{B14} & \texttt{B15} \\
\midrule
%\multicolumn{2}{l}{Avg.\ output} & 53 & 80 & 45 & 125 & 474 & 326 & 2 & 464 & 715 & 26 & 21 & 0 & 27 & 25 & 25 \\
%\midrule
 \multirow{2}{*}{\smash{\rotc{\small \em no~}}} & \multirow{2}{*}{\smash{\rotc{\small \em gen.}}} & Verdict & \vvio & \vvio & \vvio & \vvio & \vvio & \vvio & \vvio & \vvio & \vvio & \vvio & \vvio & \vsec & \vvio & \vvio & \vvio \\
   && Time [s] & 3.77 & 3.01 & 3.42 & 4.10 & 12.08 & 15.14 & 3.85 & 19.83 & 36.70 & 2.00 & 2.31 & 9.89 & 3.73 & 2.16 & 2.38 \\
\midrule 
   & \multirow{2}{*}{\smash{\rotc{\small \em gen.}}} & Verdict & \vsec & \vvio & \vsec & \vvio & \vvio & \vsec & \vvio & \vvio & \vsec & \vvio & \vvio & \vsec & \vvio & \vvio & \vvio \\
    & & Time [s] & 6.68 & 1.21 & 5.40 & 0.21 & 9.36 & 19.99 & 0.73 & 8.18 & 33.01 & 0.09 & 0.09 & 0.47 & 0.10 & 0.09 & 0.09 \\
\bottomrule
\end{tabular}%
}

\vspace{1.0ex}

\resizebox{\textwidth}{!}{%
\begin{tabular}{p{.5mm}ll*{15}{c}}
\toprule
 && & \texttt{B16} & \texttt{B17} & \texttt{B18} & \texttt{B19} & \texttt{B20} & \texttt{B21} & \texttt{B22} & \texttt{B23} & \texttt{B24} & \texttt{B25} & \texttt{B26} & \texttt{B27} & \texttt{B28} & \texttt{B29} & \texttt{B30} \\
\midrule
%\multicolumn{2}{l}{Avg.\ output} & 29 & 39 & 4 & 34 & 2 & 51 & 3 & 0 & 201 & 475 & 628 & 431 & 326 & 62 & 34 \\
%\midrule
\multirow{2}{*}{\smash{\rotc{\small \em no~}}}& \multirow{2}{*}{\smash{\rotc{\small \em gen.}}} & Verdict & \vvio & \vvio & \vsec & \vvio & \vvio & \vvio & \vsec & \vvio & \vvio & \vvio & \vvio & \vvio & \vvio & \vvio & \vvio \\
& & Time [s] & 3.38 & 3.51 & 8.50 & 3.99 & 2.11 & 4.51 & 9.79 & 8.54 & 5.17 & 4.91 & 42.69 & 18.38 & 8.48 & 3.02 & 3.21 \\
\midrule
&\multirow{2}{*}{\smash{\rotc{\small \em gen.}}} & Verdict & \vvio & \vvio & \vsec & \vvio & \vvio & \vvio & \vsec & \vvio & \vvio & \vvio & \vsec & \vvio & \vsec & \vvio & \vsec \\
& & Time [s] & 0.09 & 0.09 & 0.52 & 0.08 & 0.08 & 0.13 & 0.50 & 0.25 & 1.74 & 7.68 & 33.19 & 4.35 & 24.55 & 0.84 & 8.70 \\
\bottomrule
\end{tabular}%
}
\end{table}

% \begin{figure*}[t]
%   \centering
% \includegraphics[width=.5\textwidth]{
% 	data/scatter.pdf
% }
% %\vspace{-3mm}
% \caption{The results of monitoring GNI with and without a generator based on symbolic execution.}
% \label{fig:gni-results}
% %\vspace{-5mm}
% \end{figure*}

Table~\ref{tab:gni-results}
shows the results of executing monitors on 200 observed traces generated by fuzzing the programs.
The number of events in traces ranged from a few to a few hundred.
The monitor without the generator (rows \emph{no gen.}) is not guaranteed to give correct results, because it uses observed traces only; it serves as a baseline for the comparison.
% We include it to show how often the mismatch in verdicts happens in the considered experimental setup and also to have a baseline for comparing the performance.
%
The table shows that with the generator (rows \emph{gen.}), the monitor can answer that GNI is satisfied in cases where the monitor without a generator found a violation (incorrectly).
Also, on our benchmarks, monitoring with the generator is often faster, because the verdict can be decided quickly (for violations) and we can avoid iterating over all observed traces for the quantifier $\pi_3$.
However, monitors with the generator can also be slower (e.g., on benchmarks \texttt{B1} or \texttt{B28}), because SE may need to search many execution paths.

\section{Related Work}
\label{sec:related work}
\vspace*{-2mm}

\noindent\textbf{Hyperproperties with Generators.}
To the best of our knowledge, this work is the first to explore the use of generators as explicit semantic objects in hyperproperty specifications. 
The closest work %
is the research on the expressive power and decidability of a fragment of first-order logic using constrained and unconstrained quantifiers~\cite{flavorsFOL25}.
Constrained quantifiers correspond closely to observing quantifiers, while unconstrained quantifiers can be viewed as generators that always return all possible traces. 

Second-order hyperproperties~\cite{secondHyper23} allow one to declaratively construct sets of traces directly within the logic. 
While this enables powerful forms of reasoning, the declarative specification and the procedural aspects of trace construction are tightly coupled. 
Our approach separates these concerns: formulas in our formalism remain purely declarative, while generators are introduced as independent semantic objects whose correctness can be analyzed separately.
%
%This separation enables principled reasoning about soundness and reuse across different monitoring and verification settings.

% Merged into the paragraph on constructed executions below:
%The limitations of monitoring hyperproperties using only observed executions have already been identified in \cite{rvHyperStaticborzoo18}, where the authors propose to use static analysis to add knowledge about the monitored system to the monitor.
%%
%Later, Stucki et al.~\cite{stucki2021gray} realize this idea by implementing a gray-box monitor for distributed data minimization using a symbolic representation of a system to instantiate the existential quantifier.
%
%Their monitors instantiate some of the hyperproperty quantifiers with information previously extracted from the system, while we propose making this distinction visible at specification time.

\smallskip
\noindent\textbf{Monitoring Hyperproperties.}
Agrawal and Bonakdarpour~\cite{Agrawal16} introduced the first monitoring algorithm for a restricted class of \(k\)-safety properties expressed in HyperLTL.
Later, Brett et al.~\cite{brett17} gave a monitor for HyperLTL formulas based on formula rewriting.
Automata-based monitoring algorithms for HyperLTL were introduced by
Finkbeiner et al.~\cite{monitorHyper17,monitoringHyperLTL19} and implemented in the tool \emph{RVHyper}~\cite{RVHyper18}.
Hahn et al.~\cite{hahn19,Hahn19rv} and Aceto et al.~\cite{AcetoAAF22} presented constraint-based monitors for HyperLTL formulas.
All these solutions interact with the system as a black-box and 
monitor for \emph{synchronous} hyperproperties with no quantifier-alternation.
Chalupa and Henzinger~\cite{prefixTransducers23} introduced 
transducers for monitoring alternation-free asynchronous hyperproperties.
A different approach for monitoring asynchronous hyperproperties was presented by Beutner et al.~\cite{BeutnerRV24}, who introduced and used finite trace semantics for second-order HyperLTL.% along its monitoring algorithm.
Both works use first-order quantification over observations, while second-order HyperLTL allows quantification over sets of traces, which are, in~\cite{BeutnerRV24}, restricted to sets specified as fixpoints.

\smallskip
\noindent\textbf{Constructed Executions in Prior Monitoring and Enforcement.}
Constructed, rather than observed, executions already appear across work on runtime-verification of information-flow security and hyperproperties.
Secure multi-execution (SME)~\cite{sme10,smeFaceted16} re-runs a program on the public inputs of an observed run using default values for the secret inputs and %, under barrier synchronization,
compares the public outputs of the two runs; this re-execution is a generator in our sense.
The testable hypersafety policies of Ngo et al.~\cite{blackenforc15} are decided at runtime by \emph{test generators}, functions returning a finite, sufficient set of alternative input streams, which are likewise instances of our generators.
The gray-box monitors of Stucki et al.~\cite{stucki2021gray} instantiate existential quantifiers from a symbolic representation of the monitored system.
That work is a follow-up on the previous paper where the limitation of using only observed traces have been spelled~\cite{rvHyperStaticborzoo18}.
Our generators make this recurring pattern explicit at the specification level and let its soundness be analyzed uniformly.

%\medskip

\vspace*{-2mm}
\section{Conclusion}
\label{sec:conclusion}
\vspace*{-2mm}
We presented a new approach to monitoring hyperproperties where monitors use \emph{generators} to actively construct traces.
While a similar approach had previously been used in several works, 
we provide a general framework by distinguishing between observed and constructed traces not ad-hoc, but explicitly at the \emph{specification level}.
With constructed traces, we can monitor some hyperproperties that are not monitorable with observed traces only.
Generators can also be used to actively search for violation witnesses, boosting performance of monitors for monitorable hyperproperties.
We support these claims with experiments.
Future research includes a theoretical study of generators, their properties, and their automatic synthesis.
Also, we plan to explore efficient and effective use of more complex generators, possibly with a learning component.

\bigskip
\noindent\textbf{Acknowledgements}
This work was supported in part by the European Research Council (ERC) Grant VAMOS (No. 101020093), the Austrian Science Fund (FWF) SFB project SpyCoDe 10.55776/F85, and by the Advanced Research and Invention Agency under the Safeguarded AI programme (MSAI-PR01-P047).

%% Bibliography
\bibliographystyle{IEEEtran}
\bibliography{main}

\clearpage
\appendix

\section{Constructed Traces in Hyperproperty Monitoring}
\label{sec:intro:rel_work}

We discuss examples from the literature  that illustrate that constructed executions are a recurring semantic device in hyperproperty monitoring, to motivate generators as a unifying abstraction for reasoning about such executions.

\medskip
\noindent\textbf{Secure multi-execution (SME).}
SME  \cite{sme10,SME16}
enforces information-flow policies by executing multiple copies of a program, one for each security level. 
For two security labels $H$ and $L$, SME runs two copies of the program: $P_H$ and $P_L$.
Secret inputs are provided exclusively to $P_H$, while $P_L$ is executed with default values for secret inputs; public inputs are given to both $P_H$ and $P_L$.
Rafnsson and Sabelfeld \cite{SME16} extend SME with barrier synchronization, where the public outputs of the two executions are compared.
If the public outputs are the same, no information leakage is detected.
SME with barrier synchronization monitors for the hyperproperty:
%
%\vspace{-4mm}
\begin{align}
	\label{sme:spec}
	\obsQ{\forall}{\traceVar} \ \genQ{\forall}{\traceVar'}{P_L(\inlow(\traceVar))} \ \outlow(\traceVar) =  \outlow(\traceVar')
\end{align}
where 
$P_L(\inlow(\traceVar))$ defines the execution obtained by running the program on the public inputs observed in 
$\traceVar$
and default values for secret inputs.
In our framework, $P_L(\inlow(\traceVar))$ is a generator.
%
%Here, the generator is a copy of the system with secret values set to default.
%

\medskip
\noindent\textbf{Testable hypersafety policies.}
Ngo et al. \cite{blackenforc15} introduced \emph{testable hypersafety policies}, a class of hyperproperties whose satisfaction can be decided by observing a finite set of executions.
They prove that all such policies can be verified at runtime using
\emph{test generators}: functions that compute a finite and sufficient set of alternative input streams to be checked by the monitor.
A deterministic reactive program \cite{blackenforc15} $C$ is \emph{non-interferent w.r.t.\ a projection $\pvF{pr}$} iff
\begin{align}
	\forall I\, \forall I'\ \pvF{pr}(I)=\pvF{pr}(I') \rightarrow C(I)|_{L} = C(I')|_{L}
\end{align}
where (following the notation from \cite{blackenforc15}) quantification is over input streams, \(I\) and \(I'\).
When the test generator is a function that returns an input stream in which secret values are replaced with default values while public inputs remain unchanged, the resulting verification is equivalent to  SME.
As the name suggests, their test generators are applications of generators as introduced in this work.

\medskip
\noindent\textbf{Gray-box monitoring of hyperproperties.}
In \cite{stucki2021gray}, the authors propose a gray-box approach to hyperproperty verification in which a model approximating the system is used to predict its future behavior. 
This model effectively reduces the set of possible future observations that a monitor must consider in order to reach a verdict.
The authors implement their approach to verify distributed data minimization (DDM) at runtime, a hyperproperty with two universal quantifiers followed by two existential quantifiers. 
They further show that DDM is not monitorable in the absence of prior knowledge about the system under observation (i.e., in a black-box setting).
In their approach, they monitor DDM by
using a \emph{model extracted from the program} in the form of its symbolic execution tree.
They use this model 
to instantiate the existential quantifiers, while universal quantifiers range over observations of the system.
\rev{The model used here is another example of a generator.}

\section{Full Proofs}

\noindent\textbf{Theorem
\ref{thm:sound_obs}}
\emph{Let
$\varphi \mathop{=} \obsQ{\Quant_1}{\traceVar_1} \LLL \obsQ{\Quant_n}{\traceVar_n}\ \psi(\traceVar_1, \LLL, \traceVar_n)$
where
$\psi(\traceVar_1,\LLL,\traceVar_n)$ is a multi-trace specification.
%	be a hyper formula, $\VarTrace\mathop{=}\{\traceVar_1,\LLL,\traceVar_n\}$ 
%
%Let $\setTraces$ be a set of traces and $\funcSign$ a generators interpretation.
%
Let $\subst$ be a generators substitution such that, for every $\traceVar_i\mathop{\in}\{\traceVar_1,\LLL,\traceVar_n\}$ with $\subst[\traceVar_i]\mathop{=}\func(\overline{\traceVar})$ defined,
%\begin{align*}
$	\Quant_i\mathop{=}\exists \tAnd \func^{\funcSign}\mathop{\sqsubseteq}\setTraces,
\text{ or }
\Quant_i\mathop{=}\forall \tAnd \setTraces\mathop{\sqsubseteq}\func^{\funcSign}.$
%\end{align*}
If $(\setTraces,\funcSign) \models \varphi[\subst]$ then $\setTraces \models \varphi$.
}

\begin{proof}
Since $\varphi$ and $\varphi[\subst]$ are closed (recall that every $\traceVar_i$ is bound and $\psi$ ranges only over $\traceVar_1,\LLL,\traceVar_n$) and $\varphi$ has no generators, it suffices to show, for the empty assignment $\traceAssign_{\emptyset}$, that $(\traceAssign_{\emptyset},\setTraces,\funcSign)\models\varphi[\subst]$ implies $(\traceAssign_{\emptyset},\setTraces,\funcSign)\models\varphi$.
We will prove this statement by downward induction on the number of the quantifiers at the beginning of $\varphi$.
Before we start our proof, we define, for $0\mathop{\leq}k\mathop{\leq}n$, that  $\varphi_k\mathop{=}\obsQ{\Quant_{k+1}}{\traceVar_{k+1}}\LLL\obsQ{\Quant_n}{\traceVar_n}\ \psi$, i.e., $\varphi_k$ is the suffix of $\varphi$ after its first $k$ quantifiers, and $\varphi_k[\subst]$ the corresponding suffix of $\varphi[\subst]$.
Then, we will prove by downward induction on $k$ that, for every assignment $\traceAssign$ assigning each of $\traceVar_1,\LLL,\traceVar_k$ a pointed-trace $(\trace,0)$ with $\trace\mathop{\in}\setTraces$,
\begin{align*}
(\traceAssign,\setTraces,\funcSign)\models\varphi_k[\subst]\quad\Longrightarrow\quad(\traceAssign,\setTraces,\funcSign)\models\varphi_k.
\end{align*}
We observe that $k\mathop{=}0$ with the empty assignment gives us the theorem.

\smallskip\noindent\emph{Base case ($k\mathop{=}n$).} Here $\varphi_n[\subst]\mathop{=}\psi\mathop{=}\varphi_n$, since $\subst$ rewrites only quantifiers and $\psi$ is quantifier-free.

\smallskip\noindent\emph{Inductive step.} 
We assume by induction hypothesis (IH) that the claim holds for $k+1$, and consider an arbitrary assignment $\traceAssign$ as described above s.t.\ $(\traceAssign,\setTraces,\funcSign)\models\varphi_k[\subst]$. 
We proceed now by cases on whether $\subst$ is defined for $\traceVar_{k+1}$.

If $\subst$ is undefined for $\traceVar_{k+1}$, the leading quantifier is $\obsQ{\Quant_{k+1}}{\traceVar_{k+1}}$ (ranging over $\setTraces$) in both formulas, so the claim follows directly from the (IH).

Otherwise, if $\subst[\traceVar_{k+1}]\mathop{=}\func(\overline{\traceVar})$, by well-formedness of the formula, the trace variables in $\overline{\traceVar}$ are in $\{\traceVar_1,\LLL,\traceVar_k\}$, so $\traceAssign(\overline{\traceVar})$ is a tuple of traces in $\setTraces$. 
We go now over the cases for each possible quantifier.

\emph{$\Quant_{k+1}\mathop{=}\exists$ and $\func^{\funcSign}\mathop{\sqsubseteq}\setTraces$:} Then $\varphi_k[\subst]\mathop{=}\genQ{\exists}{\traceVar_{k+1}}{\func(\overline{\traceVar})}\ \varphi_{k+1}[\subst]$, and this entails that there exists $\trace\mathop{\in}\func^{\funcSign}(\traceAssign(\overline{\traceVar}))$ s.t.\ $(\traceAssign[\traceVar_{k+1}\mathop{\mapsto}(\trace,0)],\setTraces,\funcSign)\mathop{\models}\varphi_{k+1}[\subst]$. 
As $\traceAssign(\overline{\traceVar})$ are traces in $\setTraces$ and $\func^{\funcSign}\mathop{\sqsubseteq}\setTraces$, by Def.~\ref{def:gen:approx}, $\func^{\funcSign}(\traceAssign(\overline{\traceVar}))\mathop{\subseteq}\setTraces$, hence $\trace\mathop{\in}\setTraces$.
By (IH), $(\traceAssign[\traceVar_{k+1}\mathop{\mapsto}(\trace,0)],\setTraces,\funcSign)\models\varphi_{k+1}$, and by $\trace\mathop{\in}\setTraces$:  $(\traceAssign,\setTraces,\funcSign)\mathop{\models}\obsQ{\exists}{\traceVar_{k+1}}\ \varphi_{k+1}\mathop{=}\varphi_k$.

\emph{$\Quant_{k+1}\mathop{=}\forall$ and $\setTraces\mathop{\sqsubseteq}\func^{\funcSign}$:} 
Then $(\traceAssign[\traceVar_{k+1}\mathop{\mapsto}(\trace,0)],\setTraces,\funcSign)\mathop{\models}\varphi_{k+1}[\subst]$ for every $\trace\mathop{\in}\func^{\funcSign}(\traceAssign(\overline{\traceVar}))$. 
As $\traceAssign(\overline{\traceVar})$ are traces in $\setTraces$ and $\setTraces\mathop{\sqsubseteq}\func^{\funcSign}$, Definition~\ref{def:gen:approx} gives $\setTraces \mathop{\subseteq} \func^{\funcSign}(\traceAssign(\overline{\traceVar}))$.
Together with (IH) we have that $(\traceAssign[\traceVar_{k+1}\mathop{\mapsto}(\trace,0)],\setTraces,\funcSign)\models\varphi_{k+1}$ for all $\trace\mathop{\in}\setTraces$. Hence $(\traceAssign,\setTraces,\funcSign)\models\obsQ{\forall}{\traceVar_{k+1}}\ \varphi_{k+1}\mathop{=}\varphi_k$.
\end{proof}

\noindent\textbf{Theorem
\ref{thm:correct_ative}}
\emph{
Let
\(\varphi\) be a hyper formula where each generator appears in at most one generating quantifier.
Let \(\subst\) be a generators substitution.
Let $\setTraces$ be a set of traces,
\(\funcSign\) and \(\funcSignS\) be generators interpretations s.t.\
\((\varphi[\subst],\funcSignS)\corrFI (\varphi,\funcSign)\). 
For all assignments \(\traceAssign\):
if \((\traceAssign, \setTraces,\funcSignS) \mathop{\models} \varphi[\subst]\)
then \({(\traceAssign, \setTraces,\funcSign)\mathop{\models} \varphi}\).
}

\begin{proof}
We show, by induction on the structure of $\varphi$, that for all assignments
$\traceAssign_s$ and $\traceAssign$,
if $(\varphi[\subst],\traceAssign_s,\funcSignS)\corrFI(\varphi,\traceAssign,\funcSign)$ then
\begin{align*}
	(\traceAssign_s,\setTraces,\funcSignS)\models\varphi[\subst] \ \Longrightarrow\ (\traceAssign,\setTraces,\funcSign)\models\varphi.
\end{align*}
The theorem is the special case $\traceAssign_s\mathop{=}\traceAssign$, since
$(\varphi[\subst],\funcSignS)\corrFI(\varphi,\funcSign)$.

\medskip
\noindent\emph{Base case: $\varphi$ is a multi-trace specification.}
Then $\varphi[\subst]\mathop{=}\varphi$, and the claim follows directly by Def.~\ref{def:form_correct}.

\medskip
\noindent\emph{Inductive cases: $\varphi\mathop{=}\genQ{\Quant}{\traceVar}{\func(\overline{\traceVar})}\ \varphi'$}.
We observe that
$\varphi[\subst]\mathop{=}\genQ{\Quant}{\traceVar}{\funcS(\overline{\traceVar}_s)}\ \varphi'[\subst]$,
where $\funcS(\overline{\traceVar}_s)\mathop{=}\subst[\traceVar]$ if defined and
$\func(\overline{\traceVar})$ otherwise.
We assume $(\star)$: $(\varphi[\subst],\traceAssign_s,\funcSignS)\corrFI(\varphi,\traceAssign,\funcSign)$.

\smallskip\noindent\emph{Existential case: $\Quant\mathop{=}\exists$.}
We assume $(\traceAssign_s,\setTraces,\funcSignS)\models\varphi[\subst]$, which is equivalent to:
\begin{align*}
	\exists\, \trace_s\mathop{\in}\funcS^{\funcSignS}(\traceAssign_s(\overline{\traceVar}_s)):\ (\traceAssign_s[\traceVar\mathop{\mapsto}(\trace_s,0)],\setTraces,\funcSignS)\models\varphi'[\subst].
\end{align*}
By $(\star)$ and the existential clause of Def.~\ref{def:form_correct}, for this
$\trace_s$ there is $\trace\mathop{\in}\func^{\funcSign}(\traceAssign(\overline{\traceVar}))$ with
\begin{align*}
	(\varphi'[\subst],\traceAssign_s[\traceVar\mathop{\mapsto}(\trace_s,0)],\funcSignS)\corrFI(\varphi',\traceAssign[\traceVar\mathop{\mapsto}(\trace,0)],\funcSign).
\end{align*}
Since
$(\traceAssign_s[\traceVar\mathop{\mapsto}(\trace_s,0)],\setTraces,\funcSignS)\mathop{\models}\varphi'[\subst]$, we can apply the induction hypothesis:
\begin{align*}
	&(\traceAssign[\traceVar\mathop{\mapsto}(\trace,0)],\setTraces,\funcSign)\models\varphi' \text{ and, thus, }
	\exists\, \trace\mathop{\in}\func^{\funcSign}(\traceAssign(\overline{\traceVar})) \tSt (\traceAssign[\traceVar\mathop{\mapsto}(\trace,0)],\setTraces,\funcSign)\models\varphi',
\end{align*}
that is, $(\traceAssign,\setTraces,\funcSign)\models\varphi$.

\smallskip\noindent\emph{Universal case: $\Quant\mathop{=}\forall$.}
We assume $(\traceAssign_s,\setTraces,\funcSignS)\mathop{\models}\varphi[\subst]$.
That is, for every $\trace_s\mathop{\in}\funcS^{\funcSignS}(\traceAssign_s(\overline{\traceVar}_s))$,
$(\traceAssign_s[\traceVar\mathop{\mapsto}(\trace_s,0)],\setTraces,\funcSignS)\models\varphi'[\subst]$.
Consider arbitrary $\trace\mathop{\in}\func^{\funcSign}(\traceAssign(\overline{\traceVar}))$.
By $(\star)$ and the universal clause of Def.~\ref{def:form_correct}, there is
$\trace_s\mathop{\in}\funcS^{\funcSignS}(\traceAssign_s(\overline{\traceVar}_s))$ with
\begin{align*}
	(\varphi'[\subst],\traceAssign_s[\traceVar\mathop{\mapsto}(\trace_s,0)],\funcSignS)\corrFI(\varphi',\traceAssign[\traceVar\mathop{\mapsto}(\trace,0)],\funcSign).
\end{align*}
For this $\trace_s$ the premise gives
$(\traceAssign_s[\traceVar\mathop{\mapsto}(\trace_s,0)],\setTraces,\funcSignS)\models\varphi'[\subst]$,
so, by the induction hypothesis applied to $\varphi'$,
$(\traceAssign[\traceVar\mathop{\mapsto}(\trace,0)],\setTraces,\funcSign)\models\varphi'$.
Hence, $(\traceAssign,\setTraces,\funcSign)\models\varphi$.
\end{proof}

\noindent{\bf Lemma \ref{prop:monotonic_hyper_gen}.}
\emph{
Let \(\varphi\)
be a hyper specification with only generating quantifiers
and \(\funcSign\) an interpretation for its generators.
For all sets of traces \(\setTraces\) and \(\setTraces'\) over \(\alphabet\) s.t.\ \(\setTraces\mathop{\subseteq}\setTraces'\):
if \((\setTraces,\funcSign)\mathop{\models}\varphi\)
then
\((\setTraces',\funcSign)\mathop{\models}\varphi\).}

\begin{proof}
Our assumption \((\setTraces,\funcSign)\mathop{\models}\varphi\) is equivalent to:
\begin{align*}
	\text{exists } &\traceAssign  \text{ s.t., for all } \traceVar\mathop{\in}\Free{\varphi} \tAnd (\trace,i)\mathop{=}\traceAssign(\traceVar),
	\trace\mathop{\in}\setTraces \tAnd (\traceAssign, \setTraces, \funcSign)\models \varphi.
\end{align*}
By \(\setTraces\mathop{\subseteq}\setTraces'\), the previous implies:
\begin{align*}
	\text{exists } &\traceAssign  \text{ s.t., for all } \traceVar\mathop{\in}\Free{\varphi} \tAnd(\trace,i)\mathop{=}\traceAssign(\traceVar),
	\trace\mathop{\in}\setTraces' \tAnd (\traceAssign, \setTraces, \funcSign)\models \varphi.
\end{align*}
By \(\varphi\) not having any observing quantifier, the satisfaction relation is independent of \(\setTraces\). Hence:
\begin{align*}
	\text{exists } &\traceAssign  \text{ s.t., for all } \traceVar\mathop{\in}\Free{\varphi} \tAnd \\ &(\trace,i)\mathop{=}\traceAssign(\traceVar),
	\trace\mathop{\in}\setTraces' \tAnd (\traceAssign, \setTraces', \funcSign)\models \varphi,
\end{align*}
which is precisely \((\setTraces',\funcSign)\mathop{\models}\varphi\).
\end{proof}

\medskip
\noindent Since free trace variables are interpreted existentially, for every set of traces \(\setTraces\):
\((\setTraces,\funcSign)\mathop{\models}\neg\psi\) holds iff some assignment falsifies \(\psi\), i.e., iff \(\setTraces.\{\termSym\}\mathop{\notin}\hyperSet{\varphi_{\forall},\funcSign}\); dually, \((\setTraces,\funcSign)\mathop{\models}\psi\) holds iff \(\setTraces.\{\termSym\}\mathop{\in}\hyperSet{\varphi_{\exists},\funcSign}\). 
We use these facts freely below.

\medskip
\noindent{\bf Proposition \ref{prop:univ:goodPre}.}
\emph{
Let \(\psi\) be a hyper specification with only generating quantifiers and \(\funcSign\) an interpretation for its generators.
For \(\varphi_{\forall}\mathop{=}\obsQ{\forall}{\traceVar_1}\LLL\obsQ{\forall}{\traceVar_n}\ \psi\),
either \(\goodSet(\hyperSet{\varphi_{\forall},\funcSign})\mathop{=}\emptyset\) or
\(\hyperSet{\varphi_{\forall},\funcSign}\mathop{=}\boundPowerSet{}{\term}{\alphabet^*\{\termSym\}^*}\).
While for \(\varphi_{\exists}\mathop{=}\obsQ{\exists}{\traceVar_1}\LLL\obsQ{\exists}{\traceVar_n}\ \psi\),
either \(\badSet(\hyperSet{\varphi_{\exists},\funcSign})\mathop{=}\emptyset\) or
\(\hyperSet{\varphi_{\exists},\funcSign}\mathop{=}\emptyset\).
}
\begin{proof}
We start with the universal case.
If \(\hyperSet{\varphi_{\forall},\funcSign}\mathop{=}\boundPowerSet{}{\term}{\alphabet^*\{\termSym\}^*}\), the second disjunct holds.
Otherwise, assume towards a contradiction that
\(\goodSet(\hyperSet{\varphi_{\forall},\funcSign})\mathop{\neq}\emptyset\).
So, (\(\star\)) there exists an observation \(\observ\mathop{\in} \boundPowerSet{}{\fin}{\alphabet^*\{\termSym\}^*}\) s.t.\
for all of its terminated extensions,
\(\observ' \mathop{\in} \boundPowerSet{}{\term}{\alphabet^*\{\termSym\}^*}\) with \(\observ \hyperExtend  \observ'\): \(\observ' \mathop{\in} \hyperSet{\varphi_{\forall},\funcSign}\).
By \(\hyperSet{\varphi_{\forall},\funcSign}\mathop{\neq}\boundPowerSet{}{\term}{\alphabet^*\{\termSym\}^*}\), there exists a set of traces \(\setTraces\)
s.t.\ \((\setTraces,\funcSign)\not\models\varphi_{\forall}\), equivalently
\((\setTraces,\funcSign)\models\neg\psi\).
By Lemma~\ref{prop:monotonic_hyper_gen}, (\(\star \star\)) for all sets of traces \(\setTraces'\mathop{\supseteq}\setTraces\), \((\setTraces',\funcSign)\models\neg\psi\).
Let \(\observ'\) be an arbitrary terminated extension of \(\observ\). Then \(\observ'\mathop \cup (\setTraces . \{\termSym\})\) is also a terminated extension of \(\observ\), so by
(\(\star\)) \(\observ'\mathop \cup (\setTraces . \{\termSym\})\mathop{\in} \hyperSet{\varphi_{\forall},\funcSign}\).
However, its traces include \(\setTraces\), so by (\(\star \star\)) it satisfies \(\neg\psi\), hence \(\observ'\mathop \cup (\setTraces . \{\termSym\})\mathop{\notin} \hyperSet{\varphi_{\forall},\funcSign}\), a contradiction.
Therefore \(\goodSet(\hyperSet{\varphi_{\forall},\funcSign})\mathop{=}\emptyset\).

The existential case is dual: one shows \(\badSet(\hyperSet{\varphi_{\exists},\funcSign})\mathop{\neq}\emptyset\) forces \(\hyperSet{\varphi_{\exists},\funcSign}\mathop{=}\emptyset\), using that \((\setTraces,\funcSign)\models\psi\) is equivalent to \(\setTraces.\{\termSym\}\mathop{\in}\hyperSet{\varphi_{\exists},\funcSign}\).
\end{proof}

\medskip
\noindent{\bf Theorem \ref{thm:mon}.}
\emph{
Let \(\psi\) be a linear-time hyper specification with only generating quantifiers s.t.\ \(\Free{\psi}\mathop{=}\{\traceVar_1, \ldots, \traceVar_n\}\), and \(\funcSign\) an interpretation for its generators.
The hyper formula \(\varphi_{\forall}\mathop{=}\obsQ{\forall}{\traceVar_1} \LLL \obsQ{\forall}{\traceVar_n}\ \psi\) and \(\funcSign\)
are monitorable iff either all systems and assignments are models of \(\psi\) under \(\funcSign\)
(i.e., for all \(\traceAssign\) and \(\setTraces\): \((\traceAssign, {\setTraces},\funcSign)\mathop{\models}\psi\))
or \(\badSet(\hyperSet{\varphi_{\forall},\funcSign})\mathop{\neq}\emptyset\).
The hyper formula \(\varphi_{\exists}\mathop{=}{\obsQ{\exists}{\traceVar_1} \LLL \obsQ{\exists}{\traceVar_n}\, \psi}\) and \(\funcSign\)
are monitorable iff \(\psi\) under \(\funcSign\) has no models
(i.e., \(\hyperSet{\varphi_{\exists}, \funcSign} \mathop{=} \emptyset\))
or \(\goodSet(\hyperSet{\varphi_{\exists},\funcSign})\mathop{\neq}\emptyset\).
}

\begin{proof}
We start with the universal case: \(\varphi_{\forall}\mathop{=}\obsQ{\forall}{\traceVar_1} \ldots \obsQ{\forall} {\traceVar_n}\ \psi\).

\smallskip\noindent(\(\Leftarrow\))
If for all sets of traces \(\setTraces\) and assignments \(\traceAssign\),
\((\traceAssign, {\setTraces},\funcSign)\mathop{\models}\psi\), then
$\hyperSet{\varphi_{\forall},\funcSign}=\boundPowerSet{}{\term}{\alphabet^*\{\termSym\}^*}$ and
\(\goodSet(\hyperSet{\varphi_{\forall},\funcSign})\mathop{=}\boundPowerSet{}{\fin}{\alphabet^*\{\termSym\}^*}\); so \(\varphi_{\forall}\) is monitorable, as every observation is already a good prefix.
Now consider the case \(\badSet(\hyperSet{\varphi_{\forall},\funcSign})\mathop{\neq}\emptyset\).
By definition, there exists \(\observ\mathop{\in} \boundPowerSet{}{\fin}{\alphabet^*\{\termSym\}^*}\) s.t.\ every terminated extension \(\mathcal{V} \mathop{\in} \boundPowerSet{}{\term}{\alphabet^*\{\termSym\}^*}\), \(\observ \mathop{\hyperExtend}  \mathcal{V}\), satisfies \(\mathcal{V} \mathop{\notin} \hyperSet{\varphi_{\forall},\funcSign}\).
Fix any such terminated \(\observ \mathop{\hyperExtend}  \mathcal{V}\); then \(\mathcal{V}\notin\hyperSet{\varphi_{\forall},\funcSign}\), i.e.,
\((\mathcal{V}|_{\termSym},\funcSign)\mathop{\models}\neg\psi\), where \(\mathcal{V}|_{\termSym}\) is \(\mathcal{V}\) with the termination symbol removed.
Consider now an arbitrary observation \(\observ'\mathop{\in} \boundPowerSet{}{\fin}{\alphabet^*\{\termSym\}^*}\) and one of its terminated extensions \(\mathcal{V}'\).
For every terminated \((\mathcal{V}\mathop{\cup}\mathcal{V}')\mathop{\hyperExtend} \mathcal{W}\) we have \(\mathcal{W}|_{\termSym}\mathop{\supseteq}\mathcal{V}|_{\termSym}\), so by Lemma~\ref{prop:monotonic_hyper_gen} \((\mathcal{W}|_{\termSym},\funcSign)\mathop{\models}\neg\psi\), hence \(\mathcal{W}\mathop{\notin}\hyperSet{\varphi_{\forall},\funcSign}\).
Thus \(\mathcal{V}\mathop{\cup}\mathcal{V}'\mathop{\in}\badSet(\hyperSet{\varphi_{\forall},\funcSign})\), and \(\observ'\mathop{\hyperExtend{}} \mathcal{V}\mathop{\cup}\mathcal{V}'\).
Since \(\observ'\) was arbitrary, every observation extends to a bad prefix, so \(\varphi_{\forall}\) is monitorable.

\smallskip\noindent(\(\Rightarrow\))
Assume \(\varphi_{\forall}\) and \(\funcSign\) are monitorable.
By definition, every observation
\(\observ \mathop{\in} \boundPowerSet{}{\fin}{\alphabet^*\{\termSym\}^*}\) has an extension \(\observ' \mathop{\in} \boundPowerSet{}{\fin}{\alphabet^*\{\termSym\}^*}\), \(\observ \hyperExtend \observ'\), with
\(\observ' \mathop{\in}\goodSet(\hyperSet{\varphi_{\forall},\funcSign})  \cup \badSet(\hyperSet{\varphi_{\forall},\funcSign})\).
If \(\goodSet(\hyperSet{\varphi_{\forall},\funcSign})\mathop{\neq}\emptyset\), then by Proposition~\ref{prop:univ:goodPre} \(\hyperSet{\varphi_{\forall},\funcSign}\mathop{=}\boundPowerSet{}{\term}{\alphabet^*\{\termSym\}^*}\), i.e., all systems and assignments are models of \(\psi\) under \(\funcSign\).
Otherwise \(\goodSet(\hyperSet{\varphi_{\forall},\funcSign})\mathop{=}\emptyset\), so the extension \(\observ'\mathop{\in}\badSet(\hyperSet{\varphi_{\forall},\funcSign})\), whence \(\badSet(\hyperSet{\varphi_{\forall},\funcSign})\mathop{\neq}\emptyset\).

\smallskip
The existential case is proved analogously, with \(\goodSet\), \(\hyperSet{\varphi_{\exists},\funcSign}\mathop{=}\emptyset\), and the witness \((\mathcal{V}|_{\termSym},\funcSign)\mathop{\models}\psi\) playing the dual roles.
\end{proof}

\section{Online monitoring algorithm with generators}
\label{app:algorithm}

In this appendix section, we describe the online version of the monitoring algorithm with generators in full detail.
Recall that basic monitors are monitors that can monitor universally quantified formulas (where the quantifiers range over the same set of traces).
Also, recall that the instantiation procedure works as follows: given a formula
$\psi = Q^*_1Q^*_2...Q^*_k.\ \psi_\mathit{qf}$ where
$Q_1, Q_2, ...$ are types of quantifiers such that each $Q_i$ and $Q_{i+1}$
are different types, the top-level sub-monitor instantiates quantifiers $Q_1^*$, and for each instance of quantifiers %from $Q_1^*$
(represented as a partial trace assignment), it creates a sub-monitor for the formula $Q^*_2...Q^*_k.\ \psi_\mathit{qf}$ (propagating along the assignment of quantifiers $Q_1^*$).
The sub-monitors recursively instantiate quantifiers $Q_2$, then $Q_3$, and so on until we are left with the formula $Q_k^*.\ \psi_\mathit{qf}$ which can be monitored by a basic monitor.
Basic monitors require that the input formula is universally quantified
and if it is not, we must transform the formula using the equality $\exists \pi_1, ..., \pi_k.\ \psi_\mathit{qf} \equiv \neg(\forall \pi_1, ..., \pi_k.\ \neg\psi_\mathit{qf} )$.
That is, instead of existentially quantified formula $\exists \pi_1, ..., \pi_k.\ \psi_\mathit{qf}$, we monitor the formula $\forall \pi_1, ..., \pi_k.\ \neg\psi_\mathit{qf}$.
To preserve the equality, once we obtain the result from the monitor, we negate it.
To approach the monitors uniformly, we handle all existential quantifiers when creating sub-monitors like this.
In summary, every (sub-)monitor except basic monitors does two things: it instantiates quantifiers $Q_i$ for some $i$, and creates sub-monitors to evaluate the (possibly negated) rest of the formula.
% Once (and if) %
If all its sub-monitors are evaluated, a sub-monitor concludes with the result, negating it if its monitored sub-formula was negated.

Notice that because of the way the sub-monitors are created, each sub-monitor always instantiates quantifiers from the same set of traces: be it observations or traces generated by a function.
%This allows the generator functions to be independent of the monitor internals: when a new monitor is initialized, it gets a reference to the set of traces it uses for trace assignments and it approaches them uniformly, ignoring their origin.

\subsection{Online algorithm}

In our implementation, a generator function is an object that can be queried for sets of traces that are then passed to monitors.
These sets are initially empty and are being built (either in lock-step or even concurrently to monitors) as new events come to their input traces.
Our entire implementation works incrementally: a new event can come at any time on any unfinished trace and a new trace can be announced anytime. A new trace triggers quantifier instantiation and creating new (sub-monitors), while new events trigger stepping the basic monitors and updating generated traces.

\SetKwFunction{step}{\textnormal{\textsc{step}}}
\SetKwFunction{ehlmon}{\textnormal{\textsc{basic-mon}}}
\SetKwFunction{stepehl}{\textnormal{\textsc{step-basic}}}
\SetKwFunction{handleehl}{\textnormal{\textsc{handle-basic}}}

The full online monitoring algorithm is shown in Algorithm~\ref{alg:mon_hnl}, Algorithm~\ref{alg:step_mon}, and Algorithm~\ref{alg:step_ehl}.
The pseudo-code does not show updating and extending observed traces.
We assume that we have access to an object that represents these traces and this object is being updated independently of the monitoring code (e.g., in a separate thread).
That is, new traces and new events to known traces ``come'' in the background and we can query for the new traces and new events on traces.

Generating traces from generator functions proceeds via objects $F_{f_1}, ..., F_{f_r}$ (where $f_1, ..., f_r$ are generator functions used in the monitored formula).
These objects are created in Algorithm~\ref{alg:mon_hnl} on line~\ref{alg:mon_hnl:ln-funs}.
Given an object $F_f$ for a function $f(\pi_1, ..., \pi_k)$, we assume that it internally maintains any set of traces $f(t_1, ..., t_k)$ created during monitoring, and we denote obtaining this set as $F_f(t_1, ..., t_k)$.
Upon its creation, the set $F_f(t_1, ..., t_k)$ may be empty or not fully populated, because traces $t_1, ..., t_k$ may not be fully known yet.
Therefore, the implementation of $f$ either must wait until traces $t_1, ..., t_k$ are fully known, or, if possible, it may generate the set $F_f(t_1, ..., t_k)$ incrementally as new events to input traces come. 
Note that this is not a problem and it fits well into our setup where also the observations come step by step.
Similarly as for observed traces, we do not show updating the generated traces in the pseudo-code and we assume that the updates to $F_f(t_1, ..., t_k)$ are triggered by updates to $t_1, ..., t_k$.
If we would like to add handling these updates explicitly, we can add this code into the loop in Algorithm~\ref{alg:mon_hnl} or after line~\ref{alg:step_mon:ln-genset} in Algorithm~\ref{alg:step_mon}.

Except initializing the implementation of generator functions, Algorithm~\ref{alg:mon_hnl} handles the top-level iterations of the monitoring algorithm: there is the loop that repeatedly calls the procedure $\step(M)$ (line~\ref{alg:mon_hnl:ln-step}) which triggers the next computation step of the top-level monitor $M$ (this procedure is described in Algorithm~\ref{alg:step_mon}).
Procedure $\step$ returns either the result of the evaluation of the formula $\psi$ ($\va{false}$ or $\va{true}$), or it returns $\va{unknown}$ if the monitor has no conclusive result yet.
Monitors (except basic monitors) are represented as a triple $(\psi, \Pi, \subs)$ where $\psi$ is the (sub-)formula to monitor,
$\Pi$ is a partial trace assignment that describes already instantiated quantifiers, and $\subs$ is a set of sub-monitors that are used to monitor sub-formulas.

\SetKwInput{KwReq}{require}

\IncMargin{1em}

\begin{algorithm}[t]
  \DontPrintSemicolon
  \KwData{Closed multi-trace formula $\psi = Q_1^* Q_2^*...Q_k^*: \psi_{qf}$ with $k > 0$}
  \KwResult{$\va{false}$ or $\va{true}$ if it terminates}
  \vspace*{.8em}
  
  %\SetKwFunction{step}{\textnormal{\textsc{step}}}

  \tcp{Implementations of functions\\(passed around implicitly in the rest of the algorithm)}
  $F_{f_1}, ..., F_{f_r}$\label{alg:mon_hnl:ln-funs}\;
  \BlankLine
  \BlankLine

  \tcp{The top-level monitor. It is a triple:\\ (formula to monitor, partial trace assignment, sub-monitors)}
  $M \gets (\psi, \emptyset, \emptyset)$\label{alg:mon_hnl:ln-M}\; 

  \BlankLine
  \BlankLine
  
  \While{$\mathit{true}$} { \label{alg:mon_hnl:while}
   %\tcp{2. Step function traces}
   %\ForEach{$F_f$ in $F_{f_1}, \cdots F_{f_r}$} {
   %        update $F_f$\label{ln:fun_update}\;
   %}
    $r \gets \mathit{\step(M)}$\label{alg:mon_hnl:ln-step}\;
    \If{$r \not = \va{unknown}$} { 
      \KwRet $r$
    }
   %\BlankLine
   %\ForEach {$F \in F_{f_1},..., F_{f_r}$} {
   %    update $F$ 
   %}
  }

\caption{The main loop of the algorithm to monitor multi-trace formulas.}
\label{alg:mon_hnl}
\end{algorithm}

%%%%%%%%%%%%%%%%%%%%%%
% STEP
%%%%%%%%%%%%%%%%%%%%%%

\begin{algorithm}
  \DontPrintSemicolon
  \KwData{Monitor $m = (\psi, \Pi, submonitors)$}
  \KwResult{If it terminates, the result of evaluating $\psi$ is returned.}
  \vspace*{.8em}

  \SetKwFunction{step}{\textnormal{\textsc{step}}}
  \SetKwProg{Proc}{Procedure}{}{}

  \Proc{\step{$m$}} {
    Let $m = (Q_1^*Q_2^*...Q_k^*: \psi_\mathit{qf}, \Pi, \subs)$\;
    \BlankLine
    
    \If {$k = 1$} { \label{alg:step_mon:ln-k1}
        \KwRet $\handleehl(m)$\;
    }
    \BlankLine
    \BlankLine

    \eIf{$Q_1$ is generating} {
        let $f(\pi_1, ..., \pi_n)$ be the generator function for $Q_1$\;
        $T = F_f(\Pi(\pi_1), ..., \Pi(\pi_n))$\;
        \label{alg:step_mon:ln-genset}
    } {
        $T = $ observed traces
    }
    \BlankLine
    \BlankLine
    \If {there is a new trace $t$ in $T$ not yet considered by $m$} {
        \label{alg:step_mon:ln-newtrace}
        \tcp{Create new assignments with $t$}
        \ForEach{partial trace assignment $\Pi'$ s.t. \\
                            \quad$\Pi \subseteq \Pi'$ $\land$ %\\
                            $\Pi'(q) \in T$ for all $q \in Q_1$ $\land$ %\\
                            $t \in \Pi'(Q_1)$ } {
            \eIf {$Q_1$ is existential} {
                \label{alg:step_mon:ln-mon1}
                $m' \gets (\bar{Q}_2^*...\bar{Q}_k^*: \neg\psi_\mathit{qf}, \Pi', \emptyset)$
            } {
                $m' \gets ({Q}_2^*...{Q}_k^*: \psi_\mathit{qf}, \Pi', \emptyset)$
                \label{alg:step_mon:ln-mon2}
            }
            
            $\subs \gets \subs \cup \{m'\}$ 
            \label{alg:step_mon:ln-newtrace-end}
        }
    }
    
    \BlankLine
    \BlankLine
    \ForEach{$m'\in\subs$} {
        $r\gets \step(m')$\; \label{alg:step_mon:ln-step}
        \lIf{$r = \va{unknown}$} { \KwCont }
        \eIf{$r = \va{false}$} { 
            \tcp{conclusive result}
            \leIf{$Q_1$ is existential} {
                \KwRet $\va{true}$
                \label{alg:step_mon:ln-ret1}
            }{
                \KwRet $\va{false}$
            }
        }{
            \tcp{$r = \va{true}$}
            $\subs \gets \subs \setminus \{m'\}$
        }
    }
    \BlankLine
    \BlankLine
    {
    \tcp{Termination check}
    \If{$\subs = \emptyset$ and $T$ will have no future updates} {
            \label{alg:step_mon:ln-termination}
            \leIf{$Q_1$ is existential} {
                \label{alg:step_mon:ln-ret2}
                \KwRet $\va{false}$
            }{
                \KwRet $\va{true}$
            }
    
    }
    }
  }

\caption{Evaluating next step of a (non-basic) monitor}
\label{alg:step_mon}
\end{algorithm}

The core of the monitoring algorithm is implemented in the procedure $\step$ in Algorithm~\ref{alg:step_mon}.
This procedure takes a monitor $m = (\psi, \Pi, \subs)$, and does at most four actions:
\begin{itemize}
\item [1)] If the formula  $\psi$ has only one type of quantifiers (line~\ref{alg:step_mon:ln-k1}), then the formula (or its negation) can be monitored using a basic monitor.
Therefore, the algorithm delegates further steps to the procedure $\handleehl$, which we describe later.
\item [2)] If $k > 1$ (note that the algorithm has the requirement that $k > 0$), then the task of this monitor is to instantiate quantifiers $Q_1^*$. The algorithm obtains the set of traces $T$ that are quantified by $Q_1$ and if there is a new trace that this monitor has not considered yet, it creates new trace assignments with this trace and new sub-monitors with this trace assignment (lines~\ref{alg:step_mon:ln-newtrace}--\ref{alg:step_mon:ln-newtrace-end}).
\item [3)] If $k > 1$, the algorithm iterates through all sub-monitors of the current monitor and recursively calls $\step$ on them (line~\ref{alg:step_mon:ln-step}).
\item [4)] If all sub-monitors returned $\va{true}$, the algorithm checks if the current monitor can be terminated (line~\ref{alg:step_mon:ln-termination}).
\end{itemize}

In step 2), the algorithm creates all partial trace assignments $\Pi'$ that extend the partial trace assignment $\Pi$ which captures the instantiation of quantifiers from the parent monitor.
$\Pi'$ is required to assign traces from $T$ to all quantifiers from $Q_1^*$, and that it assigns at least one of those quantifiers with the new trace $t$ ($t\in\Pi'(Q_1)$).
Therefore, $\Pi'$ assigns traces to all quantifiers in $Q_1^*$, and also all quantifiers that precede $Q_1^*$ in the monitored formula (as it extends $\Pi$).

Once we have a trace assignment $\Pi'$, the algorithm creates the sub-monitor for the rest of the formula where we fix the trace assignment $\Pi'$ (lines \ref{alg:step_mon:ln-mon1}--\ref{alg:step_mon:ln-mon2}).
If the current quantifier type is existential, we also negate the rest of the formula: we swap existential quantifiers to universal and vice versa (the operation $\bar{Q}$) and negate the body of the formula.
This is then reflected when returning the result on lines \ref{alg:step_mon:ln-ret1} and \ref{alg:step_mon:ln-ret2} where we negate the result back.

Every new monitor $(\psi', \Pi', \emptyset)$ generated on lines  \ref{alg:step_mon:ln-mon1} or \ref{alg:step_mon:ln-mon2} is stored into $\subs$ so that the algorithm recursively calls $\step$ on the monitor on line~\ref{alg:step_mon:ln-step}.
If the monitor returns the result $\va{false}$, it means that the sub-formula $\psi'$ is not satisfied for the trace assignment $\Pi'$ (more precisely, for some trace assignments that extends $\Pi'$) and this monitor returns with a conclusive result (line~\ref{alg:step_mon:ln-ret1}).
If a sub-monitor returns $\va{true}$, we know that all trace assignments that extend $\Pi'$ satisfy the sub-formula, so we only remove this sub-monitor and continue monitoring other sub-monitors.

If we find out that the set of traces $T$ will have no updates in the future and all sub-monitors finished, we can terminate the current monitor. It may happen that this condition will never hold, though.

%%%%%%%%%%%%%%%%%%%%%%
% STEP eHL
%%%%%%%%%%%%%%%%%%%%%%

\begin{algorithm}[t]
  \DontPrintSemicolon
  \KwData{Monitor $m$}
  \KwResult{The monitor object for $\psi$}
  \vspace*{.8em}

  \SetKwProg{Proc}{Procedure}{}{}

  \Proc{\handleehl{$m$}} {
    Let $m = (Q^*: \psi_\mathit{qf}, \Pi, \subs)$\;
    \BlankLine
    
    \eIf{$Q$ is generating} {
        let $f(\pi_1, ..., \pi_n)$ be the generator function for $Q_1$\;
        $T = F_f(\Pi(\pi_1), ..., \Pi(\pi_n))$ \label{alg:step_ehl:ln-genset}
    } {
        $T = $ observed traces
    }
    \BlankLine
    \BlankLine
    
    \If {$\subs = \emptyset$} {
        \tcp{Create the basic monitor}
        \eIf {$Q$ is existential} {
            $m' \gets \ehlmon(\neg\psi_\mathit{qf}, \Pi, T)$
        } {
            $m' \gets \ehlmon(\psi_\mathit{qf}, \Pi, T)$
        }
        \BlankLine
        
        $\subs \gets \subs \cup \{m'\}$ 
    }
    \BlankLine
    let $m' \in \subs$\;
    $r\gets \stepehl(m')$\;
    \eIf{$r = \va{unknown}$} {
        \KwRet $\va{unknown}$
    }
    { 
        \leIf{$Q_1$ is existential} {
            \KwRet $\neg r$
        }{
            \KwRet $r$
        }
    }
  }

\caption{Handling basic monitors.}
\label{alg:step_ehl}
\end{algorithm}

The last piece of the algorithm is handling basic monitors, shown in Algorithm~\ref{alg:step_ehl}.
This code is called whenever function $\step$ is called for a monitor $(\psi', \Pi, \subs)$ where $\psi'$ has only a single type of quantifiers and therefore it or its negation can be monitored using a basic monitor. When the monitoring procedure for basic monitors is based on the procedure in~\cite{hypernodeMonitor24}, it runs in a potentially infinite loop. However, it is straightforward to modify it to perform the computation incrementally as the rest of our monitors do, and to return also $\va{true}$ or $\va{unknown}$ if no violation of the formula has been found (given the set of input traces is finite), or when no results have been reached yet. We call this procedure $\stepehl$.

The procedure $\handleehl$ first checks if the basic monitors has been already created or not. This check is done through the set $\subs$ that contains the monitor iff it has been created.
The monitor itself is created through calling $\ehlmon$ that takes the formula to monitor, the assignment to trace variables $\Pi$, and a set of input traces which are either observations or generated traces and that will instantiate the remaining unassigned trace variables.
Again, if the quantifier of the formula is existential, we negate the body of the formula and upon receiving a conclusive result we negate the result back.

\paragraph{Discussion}
The monitoring algorithm in essence only incrementally evaluates the input formula on the input traces.
As a result, the result of monitoring a formula must be discussed in a context: for example, if the monitor concludes with $\va{false}$ for a $\forall\exists$ formula with the first quantifier observing and the other quantifier generating,
we cannot conclude that the system is faulty unless the generating quantifier can truly over-approximate the set of necessary traces.

To see that the monitoring algorithm is correct, we briefly sketch an inductive proof. Additionally to the arguments in the main part of the paper, in the online algorithm we must consider that a sub-monitor may not terminate as it waits for new traces or events.

We take the correctness of basic monitors as granted, which gives us the base case ($k = 1$ during calling $\step$).
Now assume we have a formula $Q_1^*Q_2^*...Q_k^*.\ \psi_\mathit{qf}$ with $k > 1$ and that for formulas with smaller $k$ that algorithm correctly evaluates those formulas to $\va{true}$ or $\va{false}$ if it terminates.
For the formula $Q_1^*Q_2^*...Q_k^*.\ \psi_\mathit{qf}$, the algorithm considers all possible instances of the left-most quantifier and for each such instance it creates a sub-monitor for $Q_2^*...Q_k^*.\ \psi_\mathit{qf}$ which is correct by the induction hypothesis.
Now assume that all of the sub-monitors return $\va{true}$ or $\va{false}$.
If some of the sub-monitors returned $\va{false}$, the sub-formula $Q_2^*...Q_k^*.\ \psi_\mathit{qf}$ (correctly) evaluates to $\va{false}$ for the given traces. If $Q_1$ is universal, the algorithm returns $\va{false}$ which is the correct result as the whole formula is not satisfied. If $Q_1$ is existential, it means that $\neg(Q_2^*...Q_k^*.\ \psi_\mathit{qf})$ is $\va{false}$ and therefore $Q_2^*...Q_k^*.\ \psi_\mathit{qf}$ is $\va{true}$ which is what the algorithm returns.
We discuss that case when all sub-monitors return $\va{true}$ analogously.
In all other cases, the algorithm does not terminate, because some of the sub-monitors did not terminate.

\section{Translating Observing to Generating Quantifiers}

We define a translation from a hyper formula that may have observing quantifiers to an equivalent specification expressed using only generating quantifiers.
The translation simply rewrites all observing quantifiers (e.g., \(\exists \traceVar\ \varphi\)) into a quantifier ranging over the generator \(\funcTr\) (e.g., \(\genQ{\exists}{\traceVar}{\funcTr}\ \varphi\)) that is interpreted as \(\funcTr\mathop{=}\setTraces\).
In what follows, in all generators interpretation \(\funcSign\) the interpretation of \(\funcTr\) is fixed to \(\funcTr^{\funcSign}\mathop{=}\setTraces\), for all sets \(\setTraces\).
%
%Let \(\text{ObsVar}(\varphi)\) be the set of all trace variables in observing trace quantifiers in \(\varphi\). 
The substitution \(\subst_{\setTraces}^{\VarTrace}[\traceVar]\mathop{=}{\funcTr}\) 
assigns the generator returning $\setTraces$ to the trace variables in $\VarTrace$, and remains undefined for all other trace variables.

%The rewriting of \(\varphi\) that replaces all observing quantifiers with generating quantifiers over \(\funcTr\) is defined as:
%
%\(\act(\varphi) \mathop{=} \varphi[\subst_{\setTraces}].\)

\begin{proposition}
	\label{prop:subst_gen_equiv}
	Let 
	\(\setTraces\) be a set of traces and 
	\(\funcSign\) be a generators interpretation s.t.\ \(\funcTr^{\funcSign}\mathop{=}\setTraces\).
	Let \(\varphi\) be a hyper formula where $\VarTrace$ is the set of all trace variables in $\varphi$ bounded by an observing quantifier in $\varphi$.
	\((\traceAssign, {\setTraces},\funcSign)\! \mathop{\models}\! \varphi\) iff 
	\((\traceAssign, {\setTraces},\funcSign)\! \mathop{\models}\! \varphi[\subst_{\setTraces}^{\VarTrace}]\).
\end{proposition}
\begin{proof}
We prove by structural induction on hyper specifications.
For the base case, we observe that \(\varphi[\subst^{\VarTrace}_{\setTraces}]\mathop{=}\varphi\), hence the claim follows.
The only non-trivial inductive case is \(\varphi\mathop{=}\obsQ{\exists}{\traceVar}\ \varphi'\); the other cases follow from the induction hypothesis.
\begin{align*}
 &(\traceAssign, {\setTraces},\funcSign) \mathop{\models} \obsQ{\exists}{\traceVar}\ \varphi' 
 \quad \overset{\text{def. } \models}{\Leftrightarrow} \quad
 \text{exists } \trace\mathop{\in}\setTraces: (\traceAssign[\traceVar\mapsto(\trace,0)], {\setTraces},\funcSign) \mathop{\models} \varphi' \overset{\funcTr^{\funcSign}\mathop{=}\setTraces}{\Leftrightarrow}\\
 &
 \text{exists } \trace\mathop{\in}\funcTr^{\funcSign}: (\traceAssign[\traceVar\mapsto(\trace,0)], {\setTraces},\funcSign) \mathop{\models} \varphi'
 \overset{(IH)}{\Leftrightarrow}\\
 &
 \text{exists } \trace\mathop{\in}\funcTr^{\funcSign}\!: (\traceAssign[\traceVar\mapsto(\trace,0)], {\setTraces},\funcSign) \mathop{\models} \varphi'[\subst^{\VarTrace}_{\setTraces}]
 \overset{\text{def. } \models \tAnd \funcTr\text{ defined in } \funcSign}{\Leftrightarrow}\\
 &
 (\traceAssign, {\setTraces},\funcSign) \mathop{\models} \genQ{\exists}{\traceVar}{\funcTr}\ (\varphi'[\subst^{\VarTrace}_{\setTraces}])
 \overset{\text{def. } \subst^{\VarTrace}_{\setTraces}}{\Leftrightarrow}
 (\traceAssign, {\setTraces},\funcSign) \mathop{\models} (\obsQ{\exists}{\traceVar}\ \varphi')[\subst^{\VarTrace}_{\setTraces}]\qedhere
 \end{align*}
\end{proof}

In light of the proposition above, we can now restrict attention to hyper
formulas in which quantification ranges only over constructed executions.
In this setting, sound substitution no longer concerns replacing observed
quantifiers by generating ones, but rather replacing one interpretation of a
generator by another.

\end{document}